\documentclass[11pt,reqno,twoside]{article}

\usepackage[hang]{footmisc}
\usepackage{lipsum}

\usepackage{fixltx2e} % Makes \( \) equation style robust, among other
\usepackage{cmap} % Load before fontenc 

\usepackage[T1]{fontenc}
\usepackage[utf8]{inputenc}
\usepackage{graphicx}
\usepackage{placeins}
\usepackage{enumerate}
\usepackage{algcompatible}
\usepackage{subcaption}
\usepackage{verbatim}
\newcommand{\comments}[1]{}

\usepackage{soul}

\usepackage{setspace}

\let\counterwithin\relax  %DSA: i had to include this to be able to compile
\usepackage{lmodern} % Improved version of computer modern
\usepackage[scale=0.88]{tgheros} % Helvetica clone for sans serif font

\usepackage{bm} % boldmath must be called after the package

\usepackage{bbold}

\usepackage{amsmath,amsbsy,amsgen,amscd,amsthm,amsfonts,amssymb} 

\usepackage[centering,top=1.0in,bottom=1.0in,left=1in,right=1in]{geometry}

\usepackage{titling}
\usepackage{musicography}
\graphicspath{ {./images/} }

\usepackage[sf,bf,compact]{titlesec}

\usepackage{booktabs,longtable,tabu} % Nice tables
\usepackage[font=small,margin=30pt,labelfont={sf,bf},labelsep={space}]{caption}

\usepackage[usenames,dvipsnames]{xcolor}

\definecolor{dark-gray}{gray}{0.3}
\definecolor{dkgray}{rgb}{.4,.4,.4}
\definecolor{dkblue}{rgb}{0,0,.5}
\definecolor{medblue}{rgb}{0,0,.75}
\definecolor{rust}{rgb}{0.5,0.1,0.1}

\usepackage{url}
\usepackage[colorlinks=true]{hyperref}
\hypersetup{linkcolor=dkblue}    
\hypersetup{citecolor=rust}      
\hypersetup{urlcolor=rust}     

\usepackage[final]{microtype} 

\newtheoremstyle{myThm} % name
    {\topsep}                    % Space above
    {\topsep}                    % Space below
    {\itshape}                   % Body font
    {}                           % Indent amount
    {\sffamily\bfseries}                   % Theorem head font
    {.}                          % Punctuation after theorem head
    {.5em}                       % Space after theorem head
    {}  % Theorem head spec (can be left empty, meaning ‘normal’)

\newtheoremstyle{myRem} % name
    {\topsep}                    % Space above
    {\topsep}                    % Space below
    {}                   % Body font
    {}                           % Indent amount
    {\sffamily}                   % Theorem head font
    {.}                          % Punctuation after theorem head
    {.5em}                       % Space after theorem head
    {}  % Theorem head spec (can be left empty, meaning ‘normal’)

\newtheoremstyle{myDef} % name
    {\topsep}                    % Space above
    {\topsep}                    % Space below
    {}                   % Body font
    {}                           % Indent amount
    {\sffamily\bfseries}                   % Theorem head font
    {.}                          % Punctuation after theorem head
    {.5em}                       % Space after theorem head
    {}  % Theorem head spec (can be left empty, meaning ‘normal’)

\theoremstyle{myThm}
\newtheorem{theorem}{Theorem}[section]

\newtheorem{proposition}[theorem]{Proposition}

\theoremstyle{myRem}
\newtheorem{remark}[theorem]{Remark}

\theoremstyle{myDef}

\usepackage{fancyhdr}
\let\originalleft\left
\let\originalright\right
\renewcommand{\left}{\mathopen{}\mathclose\bgroup\originalleft}
\renewcommand{\right}{\aftergroup\egroup\originalright}

\usepackage{mathtools}
\mathtoolsset{centercolon}  % Makes := typeset correctly for definitions

\definecolor{mygreen}{rgb}{0.1,0.75,0.2}

\newcommand{\nc}{\normalcolor}

\providecommand{\mathbbm}{\mathbb} % In case we don't load bbm

\newcommand{\R}{\mathbbm{R}}

\newcommand{\G}{\mathcal{G}}

\renewcommand{\phi}{\varphi}
\newcommand{\eps}{\varepsilon}

\newcommand{\pih}{\pi_{\mbox {\tiny{\rm hyper}}}}
\newcommand{\piprior}{\pi_{\mbox {\tiny{\rm prior}}}}

\newcommand{\Nc}{\mathcal{N}}

\newcommand{\Jtp}{\J_{\mbox {\tiny{\rm TP}}}}
\newcommand{\Jp}{\J_{\mbox {\tiny{$p$}}}}

\newcommand{\g}{\,\vert\,}

\newcommand{\J}{{\mathsf{J}}}

\usepackage[font = small, margin=30pt]{caption}
\usepackage[]{algorithm}
\usepackage{algpseudocode}
\usepackage{enumerate}

\usepackage{graphicx}

\usepackage{authblk}
\usepackage[square,numbers]{natbib}
\usepackage{chngcntr}
\usepackage{mathrsfs}
\counterwithin{table}{section}
\counterwithin{algorithm}{section}

\title{Hierarchical Ensemble Kalman Methods with\\ Sparsity-Promoting Generalized Gamma Hyperpriors} 

\author{Hwanwoo Kim, Daniel Sanz-Alonso, and Alexander Strang}

\date{University of Chicago}

\vspace{.25in}

\makeatletter\@addtoreset{section}{part}\makeatother%
\numberwithin{equation}{section}

\newcommand{\upperRomannumeral}[1]{\uppercase\expandafter{\romannumeral#1}}

\begin{document}
\maketitle %  LEAVE HERE
% The command above causes the title to be displayed.

%>>>>> DELETE ALL CONTENT UNTIL "\end{document}"
% This is the body of your document.

% REQUIRED
\begin{abstract}
 This paper introduces a computational framework to incorporate flexible regularization techniques in ensemble Kalman methods for nonlinear inverse problems. The proposed methodology approximates  the \emph{maximum a posteriori} (MAP) estimate of a hierarchical Bayesian model characterized by a conditionally Gaussian prior and generalized gamma hyperpriors. Suitable choices of hyperparameters yield sparsity-promoting regularization. We propose an iterative algorithm for MAP estimation, which alternates between updating the unknown with an ensemble Kalman method and updating the hyperparameters in the regularization to promote sparsity. The effectiveness of our methodology is demonstrated in several computed examples, including compressed sensing and subsurface flow inverse problems. 
\end{abstract}

% REQUIRED
%\begin{keywords}
% ensemble Kalman methods, sparsity-promoting regularization, iterative alternating scheme, hierarchical inverse problems
%\end{keywords}

% REQUIRED
%\begin{AMS}
%  	68Q25, 35Q62, 62F15
%\end{AMS}

\section{Introduction}\label{sec:introduction}
Ensemble Kalman methods are a family of derivative-free, black-box optimization algorithms that rely on Kalman-based formulas to propagate an ensemble of interacting particles. Propagating an ensemble, rather than a single estimate, is advantageous when solving high-dimensional inverse problems with complex forward maps: the ensemble provides preconditioners and surrogate forward map derivatives to accelerate the optimization. This paper sets forth a new computational framework to leverage ensemble Kalman methods for the numerical solution of nonlinear inverse problems with flexible regularizers beyond standard $\ell_2$-penalties. In order to minimize objective functions that cannot be expressed in the usual nonlinear least-squares form, our approach introduces auxiliary variables which enable the use of ensemble Kalman methods within a reweighted least-squares procedure.
We show the effectiveness of our framework  in illustrative examples, including compressed sensing and subsurface flow inverse problems.

We will adopt a Bayesian viewpoint to derive our methodology. Specifically, the minimizer approximated by our main algorithm corresponds to the \emph{maximum a posteriori} (MAP) estimate of a hierarchical Bayesian model with conditionally Gaussian prior and generalized gamma hyperpriors.
 We propose an iterative algorithm for MAP estimation, which alternates between updating the unknown with an ensemble Kalman method and updating the hyperparameters in the regularization. The resulting method imposes sparsity while preserving the computational benefits of ensemble Kalman methods. In particular, for linear settings and under suitable conditions on the prior hyperparameters, our iterative scheme is globally convergent and gives meaningful uncertainty quantification in the large ensemble limit. Moreover, our approach retains convexity of the objective for certain nonlinear forward maps. Most importantly, numerical experiments  demonstrate the usefulness of prior hyperparameters that result in non-convex objectives but strongly promote sparsity.

\subsection{Related Work}
 Ensemble Kalman methods, overviewed in \cite{evensen2009data,chada2020iterative}, were first developed as scalable filtering schemes for high-dimensional state estimation in numerical weather forecasting \cite{evensen1994sequential,evensen1996assimilation}.
Since then, they have become popular algorithms in data assimilation, inverse problems, and machine learning. 
The papers \cite{gu2007iterative,li2007iterative,reynolds2006iterative} pioneered the development of ensemble Kalman methods for inverse problems in petroleum engineering and the geophysical sciences. Similar algorithms were introduced in \cite{iglesias2013ensemble,iglesias2016regularizing} inspired by classical regularization schemes  \cite{hanke1997regularizing}. 
Ensemble Kalman methods have grown into a rich family of computational tools for the numerical solution of inverse problems; our computational framework incorporates sparsity-promoting regularization into any of the numerous existing variants \cite{chada2020iterative}. For illustration purposes, this paper will consider two implementations based on the \emph{iterative ensemble Kalman filter} (IEKF) and the \emph{iterative ensemble Kalman filter with statistical linearization} (IEKF-SL). All necessary methodological background will be provided in Section \ref{sec:mainalgorithm} below.
Recent theoretical work on ensemble Kalman methods has established continuous-time and mean-field limits, as well as various convergence results, e.g. \cite{schillings2017analysis,blomker2019well,blomker2018strongly,chada2019convergence,herty2018kinetic,ding2019ensemble,kovachki2019ensemble,huang2022efficient}.

 Despite the effectiveness of ensemble Kalman methods in high-dimensional nonlinear inverse problems, efforts to broaden their scope by accommodating a wider range of regularization techniques are only starting to emerge. 
The paper \cite{chada2019tikhonov} introduced Tikhonov $\ell_2$-regularization and \cite{lee2021lp} generalized this idea using a transformation to turn $\ell_2$-penalties into $\ell_p$-penalties. This latter work highlights the key importance of suitably regularizing the inverse problem to achieve sparse reconstructions with ensemble Kalman methods. However, the transformation in \cite{lee2021lp} involves a term that is exponential in  $\frac{2}{p},$ which causes overflow for small $p$ and limits the choice of penalties that can be implemented in practice. Another approach for imposing sparsity-promoting regularization through thresholding was introduced in \cite{schneider2020imposing}. Cross-entropy loss \cite{kovachki2019ensemble} and
logistic loss \cite{SRPR21} have also been recently considered.

 The hierarchical prior model used to derive our methodology, along with an \emph{iterative alternating scheme} (IAS) for MAP estimation, were introduced in \cite{calvetti2020sparse} for sparse linear inverse problems. In fact, our computational framework builds on a series of articles \cite{calvetti2019hierachical,calvetti2020sparsity,calvetti2019brain,calvetti2015hierarchical} that considered (generalized) gamma hyperpriors for sparse solution of \emph{linear} inverse problems. This line of work also developed the IAS algorithm to compute the MAP estimate, which alternates between updating the unknown using conjugate gradient with early stopping and updating the regularization using closed formulas. Extensions to tackle both parameter estimation and uncertainty quantification tasks were introduced in \cite{agrawal2021variational,law2021sparse} using variational inference.  The IAS algorithm has been successfully implemented in applied linear inverse problems, including brain activity mapping from magnetoencephalography and identification of dynamics from time series data \cite{calvetti2015hierarchical,calvetti2019brain,agrawal2021variational}. These hierarchical Bayesian techniques are rooted in a broader literature on signal processing with emphasis on sparsity \cite{gorodnitsky1997sparse,daubechies2010iteratively} and are inspired by the classical reweighted least-squares algorithm \cite{green1984iteratively}.
Sparsity-promoting algorithms and models are also essential in statistical science \cite{tibshirani1996regression,carvalho2009handling}; our hierarchical approach has connections with empirical Bayes statistical methods \cite{robbins1992empirical} and with bilevel and data-driven methods for inverse problems \cite{bard2013practical,arridge2019solving}.

\nc

\subsection{Main Contributions}
The main contributions of this paper can be summarized as follows:
\begin{itemize}
    \item We introduce a flexible computational framework to incorporate regularization techniques in ensemble Kalman methods, including sparsity-promoting $\ell_p$-penalties with $p \in (0,1)$.
    \item Our framework generalizes the IAS algorithm  \cite{calvetti2018bayes, calvetti2019hierachical, calvetti2020sparse, calvetti2020sparsity} to nonlinear inverse problems.
    \item For linear or mildly nonlinear inverse problems, our framework provides uncertainty quantification; therefore, our methodology complements and generalizes variational inference techniques \cite{agrawal2021variational} that are only applicable for linear inverse problems with gamma hyperpriors.
    \item Our presentation gives a Bayesian interpretation to statistical methods that rely on $\ell_p$-penalties, such as \textit{least absolute shrinkage and selection operator} (LASSO) and adaptive LASSO \cite{tibshirani1996regression, zou2006adaptive}.
    \item We demonstrate the effectiveness of our methods in three computed examples: compressed sensing, a PDE-constrained inverse problem with convex forward map, and an elliptic inverse problem in subsurface flow. In addition to considering various regularization techniques, we present two implementations based on IEKF and IEKF-SL ensemble Kalman methods. The code used to reproduce our results can be found in \href{https://github.com/Hwanwoo-Kim/Lp-regularized-IEKF}{https://github.com/Hwanwoo-Kim/Lp-regularized-IEKF}, along with additional implementations based on Tikhonov ensemble Kalman inversion \cite{chada2019tikhonov}.  
\end{itemize}

\subsection{Outline}
Section \ref{sec: hierarchical} overviews the hierarchical model used to derive our computational framework. Section \ref{sec:mainalgorithm} introduces our methodology and discusses its theoretical underpinnings. Section \ref{sec: experiments} contains  numerical results. Section \ref{sec:Conclusions} closes.

\paragraph{Notation}
For matrix $P,$ we write $P\succ 0$ if $P$ is symmetric positive definite. For $P \succ 0,$ we denote by $\| \cdot \|^2_P := | P^{-1/2} \cdot |^2$ the squared Mahalanobis norm induced by the matrix $P,$ where $| \cdot |$ denotes the Euclidean norm. 
%\paragraph{Acronyms}
%We introduce all the acronyms which will be used throughout the paper.
%\begin{itemize}
   % \item Iterative Ensemble Kalman Filter: IEKF
   % Iterative Ensemble Kalman Filter with Statistical Linearization: IEKF-SL
   % \item Ensemble Kalman Inversion: EKI
   % \item Tikhonov Ensemble Kalman Inversion: TEKI
   % \item Iterative Alternating Scheme: IAS
   % \item Variational Iterative Alternating Scheme: VIAS
   % \item Least Absolute Shrinkage and Selection Operator: LASSO
%\end{itemize}

%\begin{itemize}
%    \item Dimension $u:$ d $\theta$; dimension $y:$ $k.$
%    \item Components of vectors: $i,$ (and $j$ if two indeces needed).
 %   \item Iteration of main algorithm: $\ell$ up to $L$ iterations. Superscripts. 
 %   \item Ensemble members $u_t^{(n)}$; ensemble size $N.$
 %   \item Iterations of EnKF (artificial time): $t$ up to $T.$
 %   \item Inverse problem: $y = \G(u) + \eps.$
%    \item $\eta:$ IAS parameter. 
%\end{itemize}

\section{Hierarchical Bayesian Model}\label{sec: hierarchical}
Consider the inverse problem of reconstructing an unknown $u \in \mathbb{R}^d$ from noisy data $y \in \mathbb{R}^k,$ related by
\begin{equation}\label{eq:inverseproblem}
y = \G(u) + \eps, \quad \quad \eps \sim \mathcal{N}(0, \Gamma).
\end{equation}
Here $\mathcal{G}: \mathbb{R}^d \to \mathbb{R}^k$ is a given forward map and $\eps$ represents a Gaussian measurement error with a given covariance $\Gamma \succ 0.$ Such problem setting arises in numerous applications in data assimilation, inverse problems,  machine learning, and statistics ---see for instance \cite{sanzstuarttaeb,chada2020iterative,AS10} and references therein.

To encode the prior belief that $u$ is sparse, we adopt the hierarchical model introduced in \cite{calvetti2020sparse}. First, let the components of $u$ be independent random variables 
\begin{equation}\label{CONDITIONAL_PRIOR_U}
    u_i \sim \Nc(0, \theta_i), \quad \theta_i>0, \quad 1 \le i \le d,
\end{equation}
with unknown variances $\theta_i.$  Then, small $\theta_i$ yields shrinkage in the estimation of the corresponding unknown component $u_i$. Equation \eqref{CONDITIONAL_PRIOR_U} defines a conditional Gaussian prior \begin{equation}\label{CONDITIONAL_PRIOR_Uv2}
    \pi(u\g\theta) \propto \frac{1}{\prod_{i=1}^d \sqrt{\theta_i}} \, \exp \left(-\frac{1}{2}\|u\|^2_{D_\theta} \right), \quad D_\theta := \text{diag}(\theta_1, \ldots, \theta_d).
\end{equation}
To modulate the level of sparsity, we set a hyperprior on $\theta$ from the generalized gamma distribution \cite{korolev2019generalized} given by
\begin{equation}\label{HYPER_PRIOR_THETA}
    \pi_{\text{hyper}}(\theta) := \prod_{i=1}^d \frac{|r|}{\Gamma(\beta)\vartheta_i^\beta}\theta_i^{r\beta-1}\text{exp}\left(-\frac{\theta_i^r}{\vartheta_i} \right) =  \frac{|r|^d}{\Gamma(\beta)^d} \prod_{i=1}^d \left(\frac{\theta_i^{r\beta -1}}{\vartheta_i^\beta}\right) \text{exp}\left(-\frac{\theta_i^r}{\vartheta_i}\right),
\end{equation}
where $\beta \in \mathbb{R}_{>0},\{\vartheta_i\}_{i=1}^d \subset \mathbb{R}_{>0}$ and $r \in \mathbb{R}\setminus \{0\}$. Notice that if $r = 1$, the hyperprior becomes a product of gamma distributions with a common shape parameter $\beta$ and an individual scale parameter $\vartheta_i$. Similarly, if $r = -1$, the hyperprior becomes a product of inverse gamma distributions with a common shape parameter $\beta$ and an individual scale parameter $\frac{1}{\vartheta_i}$. We refer to $r$ as the regularization parameter. In Section \ref{sec:mainalgorithm} we will demonstrate that the value of $r$ determines the level of regularization in the reconstruction of $u$.

Bayes's formula combines the likelihood function implied by the data model \eqref{eq:inverseproblem}
\begin{equation}
   \pi(y\g u,\theta) = \pi(y\g u) \propto \text{exp}\left(-\frac{1}{2}\|y - \mathcal{G}(u)\|_\Gamma^2 \right)
\end{equation}
together with the hierarchical prior $\piprior(u,\theta) := \pi(u \g \theta) \pih(\theta)$ defined by \eqref{CONDITIONAL_PRIOR_Uv2}--\eqref{HYPER_PRIOR_THETA} to give the posterior distribution
\begin{align*}
\pi(u, \theta \g y ) &\propto  \pi(y \g u,\theta) \piprior(u,\theta)  \\
&=   \pi(y \g u,\theta) \pi(u \g \theta) \pih(\theta)  
\propto \exp \bigl( -\J(u,\theta) \bigr),
\end{align*}
where 
\begin{equation}\label{sec2:objective}
\ooalign{
$\J(u,\theta ) := \overbrace{ \frac{1}{2} \| y - \G(u) \|_{\Gamma}^2 +  \frac{1}{2}\| u \|^2_{D_\theta}  }^{(a)} - \Bigl( r\beta - \frac32 \Bigr) \sum\limits_{i=1}^d \log \frac{\theta_i}{\vartheta_i} + \sum\limits_{i=1}^d \frac{\theta_i^r}{\vartheta_i}.$ \cr 
$\phantom{\J(u,\theta ) = \frac{1}{2} \| y - \G(u) \|^2  + \frac{1}{2}  } {\underbrace{ \phantom{ \,\,  \| u \|^2_{D_\theta}  - \Bigl( r\beta - \frac32 \Bigr) \sum_{i=1}^d \log \frac{\theta_i}{\vartheta_i} + \sum_{i=1}^d \frac{\theta_i^r}{\vartheta_i} } }_{(b)} }$ \cr 
}
\end{equation}
 The MAP estimate is the maximizer of the posterior density or, equivalently, the minimizer of the objective \eqref{sec2:objective}. The computational framework developed in the next section will minimize $\J(u,\theta)$ iteratively, alternating between minimizing the term $(a)$ with ensemble Kalman methods and minimizing the term $(b)$ using closed formulas for suitable choices of hyperparameters.

In the following section we aim to 1) propose a generic optimization method for the objective function in \eqref{sec2:objective} and build a connection with IAS \cite{calvetti2015hierarchical, calvetti2019hierachical}, which is only applicable when the forward map $\mathcal{G}$ is linear; 2) characterize a region of hyperparameter values $(r, \beta)$ for which the objective function is convex under suitable assumptions on $\mathcal{G}$; and 3) elucidate the relationship between the parameters $(r, \beta)$ and the regularization imposed on the reconstruction of $u$.

%When the forward map $\mathcal{G}$ is linear, \cite{calvetti2019hierachical} have identified a region for parameter values $(r, \beta)$, which convexifies the objective function given in \eqref{sec2:objective}. When the objective function is convex, one can minimize the function in a coordinate-wise, which has led to the \textit{iterative alternating sequential}(IAS) algorithm. 

%\begin{align}
%\begin{split}
%y|u &\sim \Nc(\G(u), \Gamma), \\
%u|\theta &\sim  \Nc(0, D_\theta), \quad D_\theta = \text{diag}(\theta), \\
% \theta_i &\sim \text{GG}(r,\beta, \vartheta_i) ,
% \end{split}
%\end{align}

\section{Iterative Alternating Ensemble Kalman Filters}\label{sec:mainalgorithm}
In this section we propose an iterative optimization method to compute the MAP estimate. Our approach attempts to minimize the objective \eqref{sec2:objective} in a block coordinate-wise fashion. The first term $(a)$ can be minimized over $u$ using ensemble Kalman methods and the second term $(b)$ can be easily minimized over $\theta$ for suitable choices of hyperparameters. Motivated by these observations, we initialize $\theta^0$ and compute iteratively the following updates:
\begin{align}\label{eq:uupdate}
\begin{split}
    u^{\ell + 1} &=  \arg \min_u \J(u, \theta^{\ell})   =   \arg \min_u \frac{1}{2} \| y - \G(u) \|_{\Gamma}^2 +  \frac{1}{2}\| u \|^2_{D_{\theta^{\ell}}}\,\,,   \\
    \theta^{\ell + 1} &= \arg \min_\theta \J(u^{\ell + 1}, \theta) =  \arg \min_\theta \frac{1}{2}\| u^{\ell + 1} \|^2_{D_\theta} -
    \Bigl( r\beta - \frac32 \Bigr) \sum\limits_{i = 1}^d \log \frac{\theta_i}{\vartheta_i} + \sum\limits_{i=1}^d \frac{\theta_i^r}{\vartheta_i}\,\,,  %\label{eq:thetaupdate}
\end{split}    
\end{align}
until a stopping criterion is satisfied.
We will demonstrate that, for suitable choices of $(r, \beta, \{\vartheta_i\}_{i=1}^d),$ the above iterative procedure agrees with minimizing the $\ell_p$-penalized objective
$$
\Jp(u): = \frac{1}{2}\|y-\mathcal{G}(u)\|_\Gamma^2 + \frac{1}{2} \sum_{i=1}^d w_i |u_i|^p, 
$$
where the regularization level $p$ is determined by $r,$ and the weight $w_i$ is determined by $p$ and the scale parameter $\vartheta_i$. In other words, we will show that the minimizer $u$ found by the iterative procedure \eqref{eq:uupdate} solves an $\ell_p$-regularized nonlinear optimization problem, provided appropriate hyperparameters $(r, \beta, \{\vartheta_i\}_{i=1}^d)$. The following two subsections describe the numerical implementation of each update rule and the choice of hyperparameters. Subsection \ref{ssec:mainalgorithms} contains the main algorithms, and Subsection \ref{ssec:convexity} provides sufficient conditions on the forward map and hyperarameters that ensure convexity of the objective \eqref{sec2:objective}.
%Section \ref{sec:theory} discusses desirable properties satisfied by the resulting methodology and Section \ref{sec: experiments} shows its empirical performance in computed examples. 

%In this section, we propose an iterative optimization method to obtain the MAP estimate for the unknown parameter $u$ in \eqref{sec2:objective}. To motivate our iterative procedure, suppose the objective function given in \eqref{sec2:objective} is convex. Under this assumption, the global minimizer of the objective function can be found by minimizing the function in a coordinate-wise fashion \cite{calvetti2019hierachical}, which in our case is greatly simplified by the ease with which the terms $(a)$ and $(b)$ in $\eqref{sec2:objective}$ can be minimized. Motivated from this observation, one can construct a sequence of $(u^{\ell}, \theta^{\ell})$ defined by
%\begin{align*}
 %   u^{\ell} &=  \arg \min_u \J(u, \theta^{\ell-1})   =   \arg \min_u \frac{1}{2} \| y - \G(u) \|_{\Gamma}^2 +  \frac{1}{2}\| u \|^2_{D_\theta^{\ell-1}}\,\,, \\
 %   \theta^{\ell} &= \arg \min_\theta \J(u^{\ell}, \theta^{\ell}) =  \arg \min_\theta \frac{1}{2}\| u^{\ell} \|^2_{D_\theta} +
 %   \Bigl( r\beta - \frac32 \Bigr) \sum\limits_{i = 1}^d \log \frac{\theta_i}{\vartheta_i} + \sum\limits_{i=1}^d \frac{\theta_i^r}{\vartheta_i}\,\,,
%\end{align*}
%and expect $(u^{\ell}, \theta^{\ell})$ to converge to the global minimum of $\J(u, \theta)$ under suitable assumptions on $\J$. In the following subsections, we illustrate each update rule in detail. 

\subsection{Updating $u$}\label{sec:update_U}
For a generic forward map $\mathcal{G}$, updating $u$ in \eqref{eq:uupdate} requires solving a nonlinear least-squares optimization problem. 
To this end, we will use ensemble Kalman methods designed to minimize Tikhonov-Phillips objectives of the form
\begin{equation}\label{sec3:tikhonov_philips_obj}
\Jtp(u) = \frac{1}{2}\|y-\mathcal{G}(u)\|^2_\Gamma + \frac{1}{2}\|u-m\|^2_P,
\end{equation}
where $ \Gamma, P \succ 0,$ and $m$ are given.
We remark, for later reference, that minimizing \eqref{sec3:tikhonov_philips_obj} can be interpreted as maximizing the posterior density with Gaussian likelihood and prior given by
\begin{align}
    \pi(y|u) &= \mathcal{N}\bigl(\G(u), \Gamma\bigr), \label{sec3:likelihood}\\
    \pi(u) & = \mathcal{N}(m,P) \label{sec3:prior}.
\end{align}
Since we are interested in the update \eqref{eq:uupdate}, we will take $m = 0$ throughout; $P = D_{\theta^{\ell}}$ will be iteratively updated in subsequent developments.

%given by
%\begin{align}
 %   -\log\pi(y|u) &\propto \frac{1}{2}\|y-\mathcal{G}(u)\|^2_\Gamma \label{sec3:likelihood}\\
 %   -\log\pi(u) &\propto \frac{1}{2}\|u-m\|^2_P \label{sec3:prior}.
%\end{align}
%Under such interpretation, $m$ represents the prior mean for $u$. Correspondingly $P$ and $\Gamma$ are covariance matrix of prior distribution and measurement errors, respectively.

Starting from an \emph{initial ensemble} $\{u_0^{(n)} \}_{n=1}^N,$ ensemble Kalman methods update the ensemble in an artificial discrete-time index $t$
$$ \{u_t^{(n)} \}_{n=1}^N \mapsto \{u_{t+1}^{(n)} \}_{n=1}^N$$
using Kalman formulas that promote fitness of the ensemble with data and with the prior distribution \eqref{sec3:prior} implied by the Tikhonov-Phillips regularization. The goals of fitting data and fitting the prior are balanced using an ensemble-based Kalman gain matrix, as well as certain additional random perturbation terms that control the long-time distribution of the ensemble in the large $N$ limit. We view the iteration subscript $t$ as a discrete-time index because the evolution of the ensemble may arise from the discretization of a system of stochastic differential equations, coupled by the ensemble-based Kalman gain. Once the ensemble reaches statistical equilibrium, we report the ensemble mean as an approximate solution of the optimization problem of interest. 
In this subsection we introduce two types of ensemble Kalman methods,
IEKF and IEKF-SL, which we employ in the update of $u$. These two algorithms differ in how they construct the Kalman gain and in how they introduce random perturbations in the ensemble update. To formulate these algorithms we need some notation. Given the ensemble $\{u_t^{(n)} \}_{n=1}^N$ at time $t$, we denote ensemble empirical means by
\begin{align*}
m_t &= \frac{1}{N} \sum_{n=1}^N u_t^{(n)}, \quad \quad g_t = \frac{1}{N} \sum_{n=1}^N \G \bigl(u_t^{(n)}\bigr),
\end{align*}
and we denote empirical covariances and cross-covariances by
\begin{align*}
P_t^{uu}  &= \frac{1}{N} \sum_{n=1}^N (u_t^{(n)} - m_t)  (u_t^{(n)} - m_t)^\top,  \quad \quad  \\
P_t^{uy} &= \frac{1}{N} \sum_{n=1}^N \bigl( u_t^{(n)} - m_t  \bigr)   \bigl( \G(u_t^{(n)}) - g_t  \bigr)^\top,  \\
P_t^{yy} &= \frac{1}{N} \sum_{n=1}^N \bigl(\G(u_t^{(n)}) - g_t \bigr)  \bigl(\G(u_t^{(n)}) - g_t \bigr)^\top.
\end{align*}
We will use repeatedly the principle of \textit{statistical linearization} ---see \cite{ungarala2012iterated,chada2020iterative}--- which we recall briefly.
Notice that if $\mathcal{G}$ is linear, i.e., $\mathcal{G}(u) = Gu$, we have 
$$
P_t^{uy} = P_t^{uu}G^\top. % \implies G = (P_t^{uy})^\top (P_t^{uu})^{-1}.
$$ 
The principle of statistical linearization is to approximate the Jacobian of a generic nonlinear map $\mathcal{G}$ at time $t$ using the above identity, namely
$$
\mathcal{G}'(u_t^{(n)})  \approx (P_t^{uy})^\top (P_t^{uu})^{-1}  =: G_t^N , \quad \quad n = 1, \ldots, N.
$$
Here and henceforth, $(P_t^{uu})^{-1}$ denotes the pseudoinverse of $P_t^{uu}.$
We next present IEKF, IEKF-SL, and a unified framework providing more insights on these algorithms.  

\subsubsection{Iterative Ensemble Kalman Filter (IEKF)}
Here we present the IEKF method introduced in \cite{chada2020iterative}. The pseudocode is given in Algorithm \ref{itenKF}. We refer to \cite{ungarala2012iterated, reynolds2006iterative} for other variants of IEKF.
\begin{algorithm}[H]
\caption{Iterative Ensemble Kalman Filter (IEKF) \label{itenKF}} %by Statistical Linearization 
%\begin{algorithmic}
	\STATE {\bf Input}: Number $T$ of iterations, step-size $\alpha,$ covariance $P\succ 0.$ \\
    \STATE {\bf Initialization}: Draw initial ensemble $u_0^{(n)}\stackrel{\text{i.i.d.}}{\sim} \Nc(0, P)$, \, $1 \le n \le N.$    \\ 
    \STATE {\bf For} $t = 0, 1, \ldots T$ {\bf do}:
    \begin{enumerate}
    \item Set $G^N_t = (P_t^{uy})^\top (P_t^{uu})^{-1}.$
    \item Update the Kalman gain $K_{t,0}^N = P_0^{uu} (G^N_t)^\top \bigl(G^N_t P_0^{uu} (G^N_t)^\top + \Gamma\bigr)^{-1}. $ 
    \item For $1\le n \le N,$ draw $y_t^{(n)} \stackrel{\text{i.i.d.}}{\sim} \Nc(y,\alpha^{-1} \Gamma).$
    \item  For $1\le n \le N,$ update
    \begin{equation*}
    u_{t+1}^{(n)} = u_t^{(n)}  + \alpha  \Bigl\{ K_{t,0}^N \big( y_t^{(n)}  - \G(u_t^{(n)}) \big) + (I - K_{t,0}^N G_t^N)  \big( u_0^{(n)} - u_t^{(n)}  \big)   \Bigr\}.
    \end{equation*}
    \end{enumerate}
    \STATE {\bf Output}: Final ensemble mean $ m_T.$ 
%\end{algorithmic}
\end{algorithm}
The Kalman gain $K_{t,0}^N$ is defined using the empirical covariance $P_0^{uu}$ of the initial ensemble, the approximated Jacobian $G_t^N$ of the forward map, and the covariance $\Gamma$ of the measurement error. In the update of ensemble members,  random perturbations are introduced only to the term which measures the discrepancy between data and the image of current ensemble members under the nonlinear map $\mathcal{G}$. Furthermore, the measure of fitness of ensemble members to the prior distribution is assessed by comparing each ensemble member with the corresponding initial ensemble member, drawn from the prior. 

Employing $P_0^{uu}$ in the construction of the Kalman gain $K_{t,0}^N$ and $u_0^{(n)}$ in the update of $u_t^{(n)}$ implicitly regularizes the ensemble by forcing all its members to remain in the linear span of the initial ensemble. This \textit{initial subspace property} holds for several ensemble Kalman methods \cite{iglesias2013ensemble, chada2019tikhonov, chada2020iterative}.
In addition, \cite{chada2020iterative} has shown that, in a linear forward map setting with step-size $\alpha = 1$, the ensemble mean computed by the IEKF algorithm converges (as $N \to \infty$) in a single step $(T=1)$  to the posterior mean with likelihood and prior given by \eqref{sec3:likelihood}--\eqref{sec3:prior}. The next three remarks discuss the choice of prior covariance, stopping criteria, and step-size with pointers to the literature. 

\begin{remark}
In ensemble Kalman methods, the prior covariance $P$ typically incorporates application-specific knowledge. For instance, the initial ensemble may be drawn from a uniform or log-normal distribution whose support reflects prior information \cite{iglesias2021adaptive, schneider2017earth}. Instead of sampling the initial ensemble, one can specify it deterministically using the first principal components of a suitable covariance model \cite{iglesias2013ensemble}. These considerations may be used to determine a suitable initialization $P = D_{\theta^0}$ for our main algorithms in Section \ref{ssec:mainalgorithms}.
\end{remark}

\begin{remark}
Instead of providing a total number $T$ of iterations, a stopping rule can be used. For instance, one could use \textit{Morozov's discrepancy principle}; continue the iteration until the discrepancy between the data and the forward mapping of $m_t$ falls below the noise level, i.e.,
$$
|y - \mathcal{G}(m_t) | \le \sqrt{\text{tr}(\Gamma}).
$$
\end{remark}
\begin{remark}
For simplicity we will consider a constant and fixed step-size $\alpha.$
The step-size can be chosen on-line with a line search method based on \textit{Wolfe}'s condition or \textit{ad hoc} procedures introducing additional hyperparameters \cite{gu2007iterative}. Non-constant and adaptive step-sizes have been employed in \cite{chada2019tikhonov, chada2019convergence,iglesias2021adaptive}. 
\end{remark}

\subsubsection{Iterative Ensemble Kalman Filter with Statistical Linearization (IEKF-SL)}
Here we present the IEKF-SL introduced in \cite{chada2020iterative}. The pseudocode is provided in Algorithm \ref{algorithm:IEKF_new}.

\begin{algorithm}[H]
	\caption{IEKF with Statistical Linearization (IEKF-SL) \label{algorithm:IEKF_new}}
	%\begin{algorithmic}
	\STATE {\bf Input}: Number $T$ of iterations, step-size $\alpha,$ covariance $P\succ 0.$ \\
\STATE {\bf Initialization}: Draw initial ensemble $u_0^{(n)}\stackrel{\text{i.i.d.}}{\sim} \Nc(0, P)$, \, $1 \le n \le N.$    \\ 
	\STATE {\bf For} $t = 0, 1, \ldots T$ {\bf do:}
	\begin{enumerate}
	\item Set $G_t^N = (P_t^{uy})^\top (P_t^{uu})^{-1}$.
		\item Update the Kalman gain $K_t^N = P (G_t^N)^\top \left(G_t^N P (G_t^N)^\top + \Gamma\right)^{-1}.$
		\item For $1\le n \le N,$ draw $y_t^{(n)} \stackrel{\text{i.i.d.}}{\sim} \Nc(y, 2\alpha^{-1} \Gamma)$, $m_t^{(n)} \stackrel{\text{i.i.d.}}{\sim} \Nc(0, 2\alpha^{-1} P).$ 
		\item For $1\le n \le N,$ set
		\begin{equation*}
			u_{t+1}^{(n)} = u_t^{(n)}+ \alpha \Bigl\{ K_t^N \big( y_t^{(n)}  - \G(u_t^{(n)}) \big) + (I - K_t^N G_t^N)  \big( m_t^{(n)} - u_t^{(n)} \big)   \Bigr\}.
		\end{equation*}
	%	or, equivalently,
	%	$$
     %   u_{t+1}^{(n)} = u_t^{(n)} +  \alpha C_t^N\Bigl((G_t^N)^\top \Gamma^{-1}(y_t^{(n)}-\mathcal{G}(u_t^{(n)})) + P^{-1}(m_t^{(n)}-u_t^{(n)}) \Bigr),\quad \, 1 \le n \le N
     %   $$
     %   where 
      %  $$
     %   C_t^N =  \left((G_t^N)^\top  \Gamma^{-1} (G_t^N) + P^{-1}\right)^{-1}.
     %   $$
	\end{enumerate}
\STATE {\bf Output}: Final ensemble mean $m_T.$
%\end{algorithmic}
\end{algorithm}
Unlike IEKF, IEKF-SL constructs the Kalman gain $K_t^N$ using the prior covariance $P$, the approximated Jacobian $G_t^N$ of the forward map, and the covariance $\Gamma$ of the measurement error. Furthermore, in the update of ensemble members, it introduces random perturbations to both the data and prior terms. The measure of fitness of ensemble members to the prior distribution is assessed by comparing each member with a perturbed prior mean.

Although IEKF-SL does not have the initial subspace property, it has been shown to achieve superior performance in a variety of inverse problems \cite{chada2020iterative}. In contrast to IEKF, the ensemble empirical mean and covariance of IEKF-SL converge, as $\alpha \to 0$ and $N, T \to \infty$, to the true posterior mean and covariance under the likelihood and prior model \eqref{sec3:likelihood}--\eqref{sec3:prior} when the forward map is linear. Therefore, for mildly nonlinear problems, IEKF-SL can be used to build approximate credible intervals for the reconstruction, allowing us to quantify uncertainties.

\begin{remark}
As in IEKF, one may consider using \textit{Morozov's discrepancy principle} and choosing the step-sizes adaptively. 
\end{remark}

\subsubsection{Unified Framework through Stochastic Differential Equations}
Both IEKF and IEKF-SL can be viewed as ensemble-based stochastic approximations of the deterministic extended Kalman filter. The extended Kalman filter finds the minimum of the objective function given in \eqref{sec3:tikhonov_philips_obj} by sequentially updating the initial guess $u_0$ according to the following rule:
$$
u_{t+1} = u_t + \alpha \Bigl( K_t  \big( y  - \mathcal{G}(u_t) \big) + (I - K_t G_t)  \big(m - u_t  \big)   \Bigr),
$$ 
where $\alpha > 0$ is a step-size, $G_t = \mathcal{G}'(u_t)$ is the Jacobian of $\mathcal{G}$, and $K_t = PG_t^\top(G_tPG_t^\top + \Gamma)^{-1}$ is the Kalman gain matrix \cite{chada2020iterative}. Setting $C_t = (I-K_tG_t)P$, we get from Woodbury's matrix inversion lemma that $K_t = C_tG_t^\top \Gamma^{-1}.$ Hence we can rewrite the preceding update rule as 
\begin{equation}\label{eq:ExKF}
    u_{t+1} = u_t + \alpha C_t\Bigl( G_t^\top \Gamma^{-1}\big( y  - \mathcal{G}(u_t) \big) + P^{-1} \big(m - u_t  \big)   \Bigr),
\end{equation}
which agrees  with a Gauss-Newton iteration applied on the Tikhonov-Phillips objective \eqref{sec3:tikhonov_philips_obj}, see \cite{bell1993iterated}. 
For small step-size $\alpha,$ one can view \eqref{eq:ExKF} as a discretization of the following differential equation, which describes the continuum version of the discrete trajectories of the iterates from the extended Kalman filter:
$$
\frac{du_s}{ds} = C_s\Bigl(\mathcal{G}'(u_s)^\top \Gamma^{-1} \bigl(y-\mathcal{G}(u_s)\bigr) + P^{-1}(m-u_s) \Bigr),
$$
where $C_s$ acts as a preconditioner. Using the identities $C_t = (I-K_tG_t)P$ and $K_t = C_tG_t^\top \Gamma^{-1}$, one can show that 
\begin{equation}\label{sec2:pre_condition_disc}
C_t^{-1} = G_t^\top  \Gamma^{-1} G_t + P^{-1},
\end{equation}
which leads, in the continuum limit, to
\begin{equation}\label{sec2:pre_condition_cont}
C_s = \bigl(\mathcal{G}'(u_s)^\top \Gamma^{-1}\mathcal{G}'(u_s) + P^{-1}\bigr)^{-1}.
\end{equation}

To give rise to an ensemble of random particles that roughly follow the continuous trajectory we defined, we employ the above deterministic differential equation as our drift term and introduce two different diffusion terms to obtain, for $1\le n \le N,$
\begin{align*}
    du^{(n)}_s &= C_s\Bigl(\mathcal{G}'(u_s)^\top \Gamma^{-1}\bigl(y-\mathcal{G}(u_s)\bigr) + P^{-1}(m-u_s) \Bigr) ds + C_s\mathcal{G}'(u_s)^\top \Gamma^{-\frac{1}{2}}dW^{(n)}_s, \\
    du^{(n)}_s &= C_s\Bigl(\mathcal{G}'(u_s)^\top \Gamma^{-1} \bigl(y-\mathcal{G}(u_s)\bigr) + P^{-1}(m-u_s) \Bigr) ds + \sqrt{2C_s}dW^{(n)}_s.
\end{align*}
Discretization of the above stochastic differential equations, together with ensemble-based approximation of the Jacobian of $\G,$ gives IEKF and IEKF-SL. First, applying Euler-Maruyama we get, for $1\le n \le N,$
\begin{align*}
    u_{t+1}^{(n)} &= u_t^{(n)} +  \alpha C_t\Bigl(G_t^\top \Gamma^{-1}\bigl(y-\mathcal{G}(u_t^{(n)})\bigr) + P^{-1}(m-u_t^{(n)}) \Bigr) + \sqrt{\alpha}
    C_t G_t^\top \Gamma^{-\frac{1}{2}}Z_t, \\
    u_{t+1}^{(n)} &= u_t^{(n)} +  \alpha C_t\Bigl(G_t^\top \Gamma^{-1} \bigl(y-\mathcal{G}(u_t^{(n)})\bigr) + P^{-1}(m-u_t^{(n)}) \Bigr) + \sqrt{\alpha}\sqrt{2C_t}Z_t,
\end{align*}
where $Z_t \sim \mathcal{N}(0, I)$. From the form of preconditioner in \eqref{sec2:pre_condition_disc}, one can show that
\begin{align*}
    \sqrt{\alpha}\sqrt{2C_t}Z_t &
\stackrel{\text{d}}{=} \sqrt{2\alpha}C_t G_t^\top \Gamma^{-\frac{1}{2}} Z_{t}^y + \sqrt{2\alpha}C_t P^{-\frac{1}{2}}Z_{t}^m,
\end{align*}
where $Z_t^y, Z_t^m \sim\stackrel{\text{i.i.d.}}{\sim} \mathcal{N}(0, I)$. By introducing randomness to $y$ through $Z_{i}^y$ and $m$ through $Z_{i}^m$, we get the following update rules, for $1 \le n \le N,$
\begin{align}
    u_{t+1}^{(n)} &= u_t^{(n)} +  \alpha C_t\Bigl(G_t^\top \Gamma^{-1} \bigl(y_t^{(n)}-\mathcal{G}(u_t^{(n)})\bigr) + P^{-1}(m-u_t^{(n)}) \Bigr), \label{sec3:first_update} \\ &~\text{with} ~y_t^{(n)} \sim \mathcal{N}(y, \alpha^{-1}\Gamma), \notag \\
    u_{t+1}^{(n)} &= u_t^{(n)} +  \alpha C_t\Bigl(G_t^\top \Gamma^{-1} \bigl(y_t^{(n)}-\mathcal{G}(u_t^{(n)})\bigr) + P^{-1}(m_t^{(n)}-u_t^{(n)}) \Bigr), \label{sec3:second_update} \\
    &~\text{with} ~y_t^{(n)} \sim \mathcal{N}(y, 2\alpha^{-1}\Gamma), m_t^{(n)} \sim \mathcal{N}(y, 2\alpha^{-1}P) \notag.
\end{align}

\begin{figure}
\centering
\includegraphics[height = 0.17\textwidth, width = 0.24\textwidth]{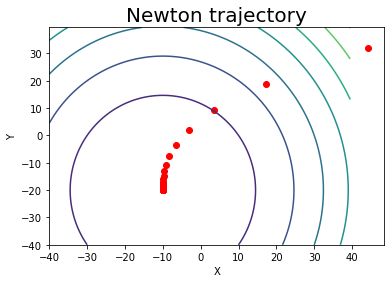}
\includegraphics[height = 0.17\textwidth,width = 0.24\textwidth]{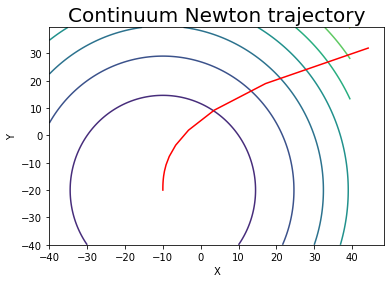}
\includegraphics[height = 0.17\textwidth,width = 0.24\textwidth]{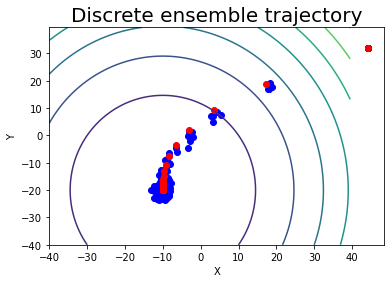}
\includegraphics[height = 0.17\textwidth,width = 0.24\textwidth]{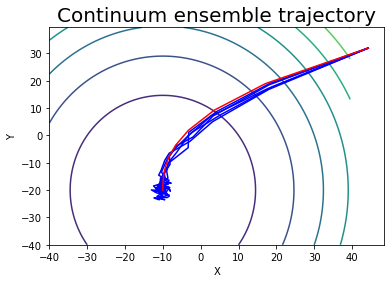}
\label{3}
\caption{Leftmost: Newton iteration. Middle-left: continuum Newton trajectory. Middle-right: ensemble Kalman iteration. Rightmost: continuum ensemble trajectory.}
\end{figure}
With random ensemble updates of the form \eqref{sec3:first_update} and \eqref{sec3:second_update}, one can avoid computing the Jacobian of $\mathcal{G}$ by using statistical linearization, which leads to a coupling of the stochastic dynamics. Doing so, we will derive IEKF and IEKF-SL.

We first consider \eqref{sec3:first_update}. If we approximate $m$ by $u_0^{(n)}$, $P$ by $P_0^{uu}$ and use statistical linearization $G_t^N$ in place of the  Jacobian $G_t$ in \eqref{sec3:first_update} we have, for $1\le n \le N,$
\begin{align*}
    u_{t+1}^{(n)} &= u_t^{(n)} +  \alpha C_{t,0}\Bigl((G_t^N)^\top \Gamma^{-1}\bigl(y_t^{(n)}-\mathcal{G}(u_t^{(n)})\bigr) + (P_0^{uu})^{-1}(u_0^{(n)}-u_t^{(n)}) \Bigr), 
\end{align*}
where
$$
C_{t,0} = \left((G_t^N)^\top  \Gamma^{-1} G_t^N + (P_0^{uu})^{-1}\right)^{-1},
$$
which leads to the IEKF scheme in Algorithm \ref{itenKF}. 

For IEKF-SL, we consider \eqref{sec3:second_update} and get, for $1\le n \le N,$
$$
u_{t+1}^{(n)} = u_t^{(n)} +  \alpha C_t^N\Bigl((G_t^N)^\top \Gamma^{-1}\bigl(y_t^{(n)}-\mathcal{G}(u_t^{(n)})\bigr) + P^{-1}(m_t^{(n)}-u_t^{(n)}) \Bigr),
$$
where 
$$
C_t^N =  \Bigl((G_t^N)^\top  \Gamma^{-1} (G_t^N) + P^{-1}\Bigr)^{-1},
$$
which leads to the IEKF-SL scheme in Algorithm \ref{algorithm:IEKF_new}.

\subsection{Updating $\theta$: Generalized Gamma and $\ell_p$-Regularization}\label{sssec:lp}
Once we solve the optimization problem for $u^{\ell + 1}$, the proposed coordinate-wise minimization strategy updates $\theta$  by setting
\begin{equation}\label{sec3:theta_update_objective}
\theta^{\ell + 1} =  \arg \min_\theta \frac{1}{2}\| u^{\ell + 1} \|^2_{D_\theta} +
\Bigl( r\beta - \frac32 \Bigr) \sum\limits_{i = 1}^d \log \frac{\theta_i}{\vartheta_i} + \sum\limits_{i=1}^d \frac{\theta_i^r}{\vartheta_i}.
\end{equation}
In this subsection we demonstrate an update rule for $\theta$ based on a particular choice of hyperparameter values in the hyperprior $\pi_{\text{hyper}}(\theta)$ in  \eqref{HYPER_PRIOR_THETA}. Specifically, we focus on $r, \beta$ satisfying $r\beta = \frac{3}{2}.$ Then, the general form of the objective function for $\theta$ in \eqref{sec3:theta_update_objective} becomes
\begin{equation}\label{gen_gamma_objective}
 \frac{1}{2}\|u\|_{D_\theta}^2  + \sum_{i=1}^d \frac{\theta_i^r}{\vartheta_i}.
\end{equation}
In order to minimize this function with respect to $\theta_i,$ we can restrict our attention to minimizing
\begin{equation}\label{sec3:theta_ind_update_obj}
\frac{u_i^2}{2\theta_i}  + \frac{\theta_i^r}{\vartheta_i}.
\end{equation}
For $r > 0$, one can observe that 
$$
\lim_{\theta_i \to 0} \left(\frac{u_i^2}{2\theta_i}  + \frac{\theta_i^r}{\vartheta_i}\right) = \infty \quad \,  ~\text{and}~ \quad \,  \lim_{\theta_i \to \infty} \left(\frac{u_i^2}{2\theta_i}  + \frac{\theta_i^r}{\vartheta_i}\right) = \infty.
$$
Therefore, the update for $\theta_i$ amounts to solving for the first order optimality condition
$$
-\frac{1}{2}\frac{u_i^2}{\theta_i^2} + \frac{r}{\vartheta_i}\theta_i^{r-1} = 0,
$$
which leads to the following update rule
\begin{equation}\label{outer_update}
\theta_i = \frac{\vartheta_i^{\frac{1}{r+1}}}{(2r)^{\frac{1}{r+1}}}|u_i|^\frac{2}{r+1}.
\end{equation}
For $r = 1$ the update precisely agrees with the update rule for IAS \cite{calvetti2019hierachical} with the sparsity parameter $\eta = 0$ and $i$th scale parameter $\theta_i^* = \frac{1}{\vartheta_i}$ in their notation. 

Plugging the update rule \eqref{outer_update} back to the objective \eqref{gen_gamma_objective}, we get
$$
\sum_{i=1}^d  \left(\frac{u_i^2}{2\theta_i} + \frac{\theta_i^r}{\vartheta_i^r}  \right) = \hspace{-0.1cm} \frac{r+1}{(2r)^{\frac{r}{r+1}}}\sum_{i=1}^d \frac{1}{\vartheta_i^{\frac{1}{r+1}}}|u_i|^{\frac{2r}{r+1}}.
$$
By setting $p := \frac{2r}{r+1} \in (0, 2)$, the objective function given in \eqref{sec2:objective} becomes
\begin{equation}\label{sec2:low_objective}
\Jp(u) : = \frac{1}{2}\|y-\mathcal{G}(u)\|_\Gamma^2 + C_r \sum_{i=1}^d w_{i,r} |u_i|^{p}, ~C_r = \frac{r+1}{(2r)^{\frac{r}{r+1}}}, ~w_{i, r} = \frac{1}{\vartheta_i^{\frac{1}{r+1}}},
\end{equation}
which can be viewed as an  $\ell_p$-regularized problem for any $p \in (0, 2)$. By adjusting $r$, one can impose different types of regularization. In particular, for $r = 1$, we have an $\ell_1$-regularized problem, whose natural Bayesian interpretation sets the gamma distribution with shape parameter $\beta = \frac{3}{2}$ and scale parameter $\vartheta_i$ as the prior for $\theta_i$. Furthermore, prior knowledge on the units/scales of each $u_i$ is captured by the $w_{i,r}$ terms, determined by the hyperparameters $\vartheta_i$ and the regularization parameter $r$. Such hyperprior-based component reweighting resembles the well-known adaptive LASSO \cite{zou2006adaptive}, which includes component-wise weights when solving $\ell_1$-minimization LASSO problems to remove the bias induced from the $\ell_1$-regularization term. In contrast to adaptive LASSO, our reweighting is motivated as a step towards finding the MAP estimate of a hierarchical Bayesian model.

%\as{Note: we might want to be a little careful here about the Bayesian interpretation. Fixing $\theta$ to the manifold where it maximizes the posterior given $u$ leads to the $\ell_p$ regularizer. But this restriction only applies to the MAP estimator. So we should be clear that, while the MAP estimator when $\beta = 3/2$ matches the MAP estimator for the corresponding prior $\exp(-\|u\|_{p}^p)$, the marginal distribution for $u$, and thus the uncertainty encoded in the prior, will be different.}

\subsection{Main Algorithms}\label{ssec:mainalgorithms}
Our proposed iterative methodology combines the algorithmic ideas introduced in Subsections \ref{sec:update_U} and \ref{sssec:lp} to compute the $u$ and $\theta$ updates in \eqref{eq:uupdate}. The procedure is summarized in Algorithm \ref{alg_1}.
%\begin{align*}
 %   u^{\ell + 1} &=  \arg \min_u \J(u, \theta^{\ell})   =   \arg \min_u \frac{1}{2} \| y - \G(u) \|_{\Gamma}^2 +  \frac{1}{2}\| u \|^2_{D_{\theta^{\ell}}},  \\
%    \theta^{\ell + 1} &= \arg \min_\theta \J(u^{\ell + 1}, \theta) =  \arg \min_\theta \frac{1}{2}\| u^{\ell + 1} \|^2_{D_\theta} + \sum\limits_{i=1}^d \frac{\theta_i^r}{\vartheta_i},  %\label{eq:thetaupdate}
%\end{align*}
%whose detailed procedure is provided in Algorithm \ref{alg_1}.

\begin{algorithm}[H]
\caption{$\ell_p$-IEKF and $\ell_p$-IEKF-SL \label{alg_1}}
%\begin{algorithmic}
\STATE   {\bf Input}: Initial $\theta^{0}$, hyperparameters $r$ and $\{\vartheta_i\}_{i=1}^d$, step-size $\alpha,$ number $T$ of inner iterations. \\
\STATE {\bf Iterate} (outer iteration) until convergence: 
\begin{enumerate}
   \item Update $u^{\ell + 1}:= m_T$ running IEKF/IEKF-SL (inner iteration) with a step size $\alpha$ and an initial covariance $P:=
   D_{\theta^{\ell}}. $
    \item Update $\theta_i^{\ell + 1} := \frac{\vartheta_i^{\frac{1}{r+1}}}{(2r)^{\frac{1}{r+1}}}|u_i^{\ell + 1}|^p$, where $p := \frac{2}{r+1}$.
    \item $\ell \rightarrow \ell+1.$
\end{enumerate}
%\end{algorithmic}
\end{algorithm}
The two coordinate-wise optimization steps serve two distinct purposes. When optimizing for $u$, the reconstruction of the unknown is updated with a given regularization; when optimizing for $\theta,$ the regularization is updated. The algorithm involves two nested iterations. First, each $u$ update runs IEKF/IEKF-SL for $T$ iterations. Second, $u$ and $\theta$ are iteratively updated, in alternating fashion, until convergence. We call the first type of iteration \emph{inner iteration} and the latter type \emph{outer iteration}. As outer iterations update $\theta$, we will also refer to them as outer regularization iterations.

The proposed method utilizes an ensemble to minimize 
\begin{equation}\label{sec33:obj}
\J(u,\theta) = \frac{1}{2}\|y-\mathcal{G}(u)\|_\Gamma^2 + \frac{1}{2}\|u\|_{D_\theta}^2  + \sum_{i=1}^d \frac{\theta_i^r}{\vartheta_i}
\end{equation}
 in a coordinate-wise fashion to obtain the minimum of the lower-dimensional objective 
\begin{equation}\label{sec33:objp}
\Jp(u) = \frac{1}{2}\|y-\mathcal{G}(u)\|_\Gamma^2 + C_r \sum_{i=1}^d w_{i, r}|u_i|^{p}.
\end{equation}
One can show that $\J(u,\theta)$ agrees with $\Jp(u)$ along the manifold 
$$
\left\{(u, \theta): \theta_i = \frac{\vartheta_i^{\frac{1}{r+1}}}{(2r)^{\frac{1}{r+1}}}|u_i|^\frac{2}{r+1}\right\}.
$$
Hence by  element-wise optimizing $\J(u,\theta)$,  we can recover the minimizer of $\Jp(u)$. 
%Introducing the auxiliary variables $\theta$ reduces the problem of optimizing the nonlinear objective $\Jp(u)$ to solving a sequence of iteratively-reweighted nonlinear least-squares problems.
%\footnote{\as{This section is pretty redundant. What new material is introduced in this paragraph that was not discussed before?}}

Algorithm \ref{alg_1} contains the pseudocode for $\ell_p$-IEKF and $\ell_p$-IEKF-SL. Compared to the method proposed in \cite{lee2021lp}, which leads to an overflow when using $\ell_p$-penalties for $p < 0.7$, our method can accommodate penalties with smaller $p$ values that promote sparsity more strongly.  For $r \in [1,2)$ under a linear forward map setting, the objective function in \eqref{sec33:obj} is globally convex, guaranteeing the existence and uniqueness of a global minimizer of \eqref{sec33:objp}. For $r \in (0,1)$, the objective function in \eqref{sec33:obj} may not be convex thus may admit multiple local minimizers, even for linear forward maps. Sufficient conditions for convexity are stated and proved in Subsection \ref{ssec:convexity}. The practical efficiency of $\ell_p$-IEKF and $\ell_p$-IEKF-SL for $r \in (0,1)$ is demonstrated in Section \ref{sec: experiments}. 

 \begin{remark}
In contrast to IEKF, $\ell_p$-IEKF partially preserves the initial subspace property: the output of $\ell_p$-IEKF lies in the span of the initial ensemble of the last $u$ update step, rather than the span of the initial ensemble of the first $u$ update step. Outer regularization iterations can be viewed as a way to adaptively modify the prior covariance for the $u$ update to reflect the sparse structure of the true parameter. As each outer iteration modifies the prior covariance, the initial subspace for each $u$ update step will change accordingly.  
\end{remark}

 \begin{remark}
For the $u$ update step, one may consider using other ensemble Kalman methods such as \textit{Ensemble Kalman Inversion} (EKI) \cite{iglesias2013ensemble} or \textit{Tikhonov Ensemble Kalman Inversion} (TEKI) \cite{chada2019tikhonov}. In contrast to IEKF and IEKF-SL, for linear forward maps the ensembles obtained using EKI and TEKI collapse to a single point in the long time limit \cite{chada2020iterative}. Consequently, EKI and TEKI do not provide uncertainty quantification, even in linear or mildly nonlinear settings, unless suitably stopped. The algorithm IEKF-SL was designed so that in linear settings the empirical covariance of the ensemble approximates the true posterior covariance in the long time asymptotic \cite{chada2020iterative}. 
\end{remark}

\begin{remark}
To determine when to terminate the outer iteration for $\ell_p$-IEKF and $\ell_p$-IEKF-SL, one may monitor the relative change of iterates. For instance, one may set a small tolerance $\tau > 0$ and terminate if
$$
\frac{\|u^{\ell+1} - u^\ell \|_{\infty}}{\|u^\ell\|_\infty} < \tau. 
$$ 
\end{remark}

\begin{remark}
Although we have focused on a particular choice of hyperparameter values, namely $r\beta = \frac{3}{2}$, the general framework extends beyond this choice  to $r \beta > \frac{3}{2} + \delta$ for arbitrary $\delta$, and to alternative $r$. In those cases, the update function which maximizes $\theta$ given $u$ is expressed implicitly as the solution to an initial value problem which can be easily solved to update $\theta$ \cite{calvetti2020sparse}. One can also consider $r = -1$, which produces effective penalty terms corresponding to prior distributions whose tails decay as power laws in $u$.  For $r= -1$  ---which corresponds to imposing an inverse gamma prior--- the objective function for the $\theta_i$ update is given by
%$$
%\frac{u_i^2}{2\theta_i} + \left(\beta + \frac{3}{2}\right) \log\frac{\theta_i}{\vartheta_i} + \frac{1}{\theta_i\vartheta_i}.
%$$
%Notice that 
%\begin{align*}
%    \lim_{\theta_i \to 0} \frac{u_i^2}{2\theta_i} + \left(\beta + \frac{3}{2}\right) \log\frac{\theta_i}{\vartheta_i} + \frac{1}{\theta_i\vartheta_i} = \infty, \\ \lim_{\theta_i \to \infty} \frac{u_i^2}{2\theta_i} + \left(\beta + \frac{3}{2}\right) \log\frac{\theta_i}{\vartheta_i} + \frac{1}{\theta_i\vartheta_i} = \infty.
%\end{align*}
%As in the $L^p$-prior setting in Subsection \ref{sssec:lp}, we deduce that the minimizer $\theta_i$ satisfies 
\begin{equation*}\label{sec3:student_update}
\theta_i = \frac{1}{\beta + \frac{3}{2}}\left(\frac{u_i^2}{2} + \frac{1}{\vartheta_i} \right).
\end{equation*}
Assuming $\vartheta_i = 1$ for all $i = 1, \ldots, d$ with $\kappa = \beta + \frac32$, the prior component for $u$ is given by 
$$
\pi(u) \propto ~\text{exp}\left(-\sum_{i=1}^d \log(u_i^2+2)^\kappa\right) = \prod_{i=1}^d \frac{1}{(u_i^2 + 2)^\kappa}.
$$
In the limit as $\beta \to 0$, i.e. $\kappa \to \frac{3}{2}$, the prior distribution for each component $u_i$ approaches the Student distribution with two degrees of freedom, a heavy-tailed distribution which favors outliers. 
\end{remark}

\subsection{Convexity}\label{ssec:convexity}
Proposition \ref{sec4:convexity_prop} below gives conditions on the forward map $\G$ and the parameters $(r, \beta)$ that ensure convexity of the objective function in \eqref{sec2:objective}. For convex objectives, our methodology can be viewed as an ensemble approximation of a coordinate descent scheme that is globally convergent under mild assumptions \cite{tseng2001convergence}.

\begin{proposition}\label{sec4:convexity_prop}
Let $\beta, r > 0.$ The following holds:
\begin{enumerate}
    \item If $r \ge 1$  or $r \le 0$, $\eta = r\beta - \frac{3}{2} \ge 0$ and $\sum_{i=1}^n \mathcal{G}_i(u) \nabla^2 \mathcal{G}_i(u) \succcurlyeq 0$, then the objective function $\J(u, \theta)$ in \eqref{sec2:objective} is convex everywhere.
    \item If $0 < r < 1$, $\eta = r\beta - \frac{3}{2}  \ge 0$ and $\sum_{i=1}^n \mathcal{G}_i(u) \nabla^2 \mathcal{G}_i(u) \succcurlyeq 0$, then the objective function $\J(u, \theta)$ in \eqref{sec2:objective} is convex provided that, for all  $i \in \{1, \ldots, d\},$
    \begin{equation}\label{convex_condition}
    \frac{\theta_i}{\vartheta_i} \le \left(\frac{\eta}{r(1-r)} \right)^{\frac{1}{r}}.
    \end{equation}
\end{enumerate}
\end{proposition}
\begin{proof}
The Hessian of $\J(u,\theta)$ is given by
$$
H = \nabla^2 \J(u,\theta) = \begin{bmatrix}
\nabla_u \nabla_u \J(u,\theta) & \nabla_\theta \nabla_u \J(u,\theta) \\
\nabla_u \nabla_\theta \J(u,\theta) & \nabla_\theta \nabla_\theta \J(u,\theta)
\end{bmatrix},
$$
where
\begin{align*}
    \nabla_u \nabla_u \J(u,\theta) &= \nabla \mathcal{G}\nabla \mathcal{G}^\top + \sum_{i=1}^n \mathcal{G}_i(u) \nabla^2 \mathcal{G}_i(u) + D_{\theta}^{-1}, \\
    \nabla_u \nabla_\theta \J(u,\theta) &= \nabla_\theta \nabla_u \J(u,\theta) = \text{diag}\left(-\frac{u_i}{\theta_i^2}\right), \\
    \nabla_\theta \nabla_\theta \J(u,\theta) &= \text{diag}\left(\frac{u_i^2}{\theta_i^3} + \frac{r(r-1)}{\vartheta_i^2}\left(\frac{\theta_i}{\vartheta_i}\right)^{r-2} + \frac{\eta}{\theta_i^2}\right).
\end{align*}
For any vector $q = \begin{bmatrix} v \\ w \end{bmatrix}$, we have
\begin{alignat*}{2}
 q^\top Hq &= \|\nabla \mathcal{G}^\top v\|^2 && \hspace{-4cm}+ \sum_{i=1}^n \mathcal{G}_i(u) (v^\top \nabla^2 \mathcal{G}_i(u) v) \\
&  && \hspace{-4cm}+ \sum_{i=1}^d \frac{v_i^2}{\theta_i} + \sum_{i=1}^d \left(\frac{u_i^2}{\theta_i^3} + \frac{r(r-1)}{\vartheta_i^2}\left(\frac{\theta_i}{\vartheta_i}\right)^{r-2} + \frac{\eta}{\theta_i^2} \right)w_i^2 -2\sum_{i=1}^d \frac{u_i}{\theta_i^2}v_iw_i   \\
 &= \|\nabla \mathcal{G}^\top v\|^2 + \sum_{i=1}^n \mathcal{G}_i(u) (v^\top \nabla^2 \mathcal{G}_i(u) v) \\ 
 & && \hspace{-4cm}+ \sum_{i=1}^d \frac{1}{\theta_i}\left(v_i - \frac{u_iw_i}{\theta_i} \right)^2 + \sum_{i=1}^d \left(\frac{r(r-1)}{\vartheta_i^2}\left(\frac{\theta_i}{\vartheta_i} \right)^{r-2} + \frac{\eta}{\theta_i^2} \right)w_i^2.
\end{alignat*}
From the assumption, the first three terms are always non-negative and the remaining term is non-negative if, for all  $i \in \{1, \ldots, d\},$
$$
\frac{r(r-1)}{\vartheta_i^2}\left(\frac{\theta_i}{\vartheta_i} \right)^{r-2} + \frac{\eta}{\theta_i^2} \ge 0,
$$
which implies the conditions of the two different cases.
\end{proof}

\section{Numerical Experiments} \label{sec: experiments} 
In this section we demonstrate the effectiveness of the proposed methodology in three examples: 1) underdetermined linear inverse problem; 2) nonlinear inverse problem with an explicit forward map that gives a convex objective for certain hyperparameter values; and 3) nonlinear elliptic inverse problem. For all three examples we assumed that only a few components of the unknown $u$ are nonzero and compared $\ell_{0.5}$/$\ell_1$-IEKF/IEKF-SL with the vanilla IEKF/IEKF-SL. Throughout, the step-size of ensemble Kalman methods is set to be $ \alpha = 0.5$ and the scale parameters are set to be $\vartheta_i = 1$ for all $1\le i \le d.$ These choices suffice to illustrate the successful regularization achieved by our method when compared to vanilla ensemble Kalman methods with $\ell_2$-regularization.
In order to clearly compare different algorithms and regularization techniques, we report the evolution of the corresponding ensembles instead of using a stopping criterion.

\subsection{Linear Inverse Problem}
Our first example explores the performance of our methods in two tasks: point estimation and uncertainty quantification. For the point estimation task, we provide comparisons with the \emph{iterative alternating scheme} (IAS) \cite{calvetti2019hierachical}, \emph{least absolute shrinkage and selection operator} (LASSO) \cite{tibshirani1996regression}, and \emph{Tikhonov ensemble Kalman inversion} (TEKI) \cite{chada2019tikhonov}. For the uncertainty quantification task, we compare credible intervals constructed using the empirical distribution of the ensembles produced by our algorithm with credible intervals constructed using the \emph{variational iterative alternating scheme} (VIAS) \cite{agrawal2021variational}. 
\subsubsection{Setting}
Consider the linear inverse problem 
$$
y = Gu + \eps, \quad \quad \eps \sim \Nc(0, 0.01I_{30}),
$$
where each component of $G \in \mathbb{R}^{30\times 300}$ is independently sampled from the standard Gaussian distribution. We assume that the true parameter $u \in \mathbb{R}^{300}$ has four nonzero components. Our goal is to recover $u$ from $y \in \mathbb{R}^{30}$. 

\begin{figure}[H]
\centering
\includegraphics[height = 0.275\textwidth]{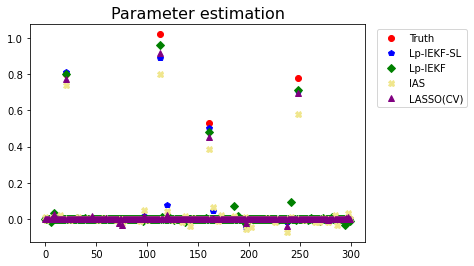}
\includegraphics[height = 0.275\textwidth]{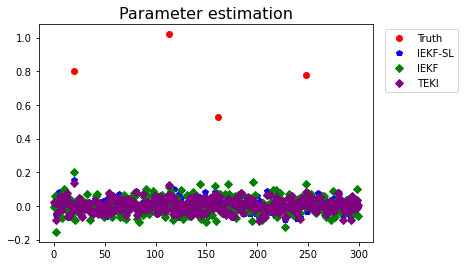}
\includegraphics[height = 0.3\textwidth, width = \textwidth]{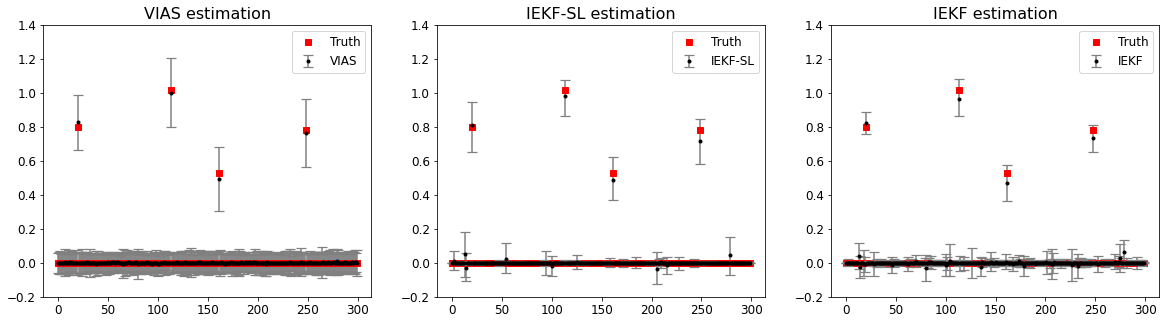}
\caption{Parameter estimation and uncertainty quantification in linear example.  Top row: parameter estimation. Bottom row: uncertainty quantification via approximate credible intervals.}
\label{EXAMPLE1_1}
\end{figure}

\subsubsection{Numerical Results}
For point estimation, we compare results obtained with $\ell_{0.5}$-IEKF, $\ell_{0.5}$-IEKF-SL, IAS, and LASSO. For both $\ell_{0.5}$-IEKF and $\ell_{0.5}$-IEKF-SL, a total of $N=300$ ensemble members were used with $T=30$ inner iterations and 10 outer iterations. The particles $u_0^{(n)}$ of the initial ensemble were independently sampled from a centered Gaussian with covariance $0.1 I_{300}$. The sparsity parameter of IAS, denoted by $\eta$ in \cite{calvetti2019hierachical}, was set to be zero. For LASSO, the regularization coefficient was chosen based on 10-fold cross-validation.

In addition to obtaining a point estimate for the parameter of interest, one can build approximate credible intervals based on the ensemble members employed in IEKF/IEKF-SL. Like in Markov chain Monte Carlo, after enough iterations ensemble members serve as good proxies for samples from the true posterior distribution in linear or mildly nonlinear settings. From these ensemble members one can obtain 2.5/97.5th sample percentiles to construct approximate 95 percentile credible intervals for each component of the parameter. We provide such approximate credible intervals for $\ell_{0.5}$-IEKF/IEKF-SL with $N=300$ ensemble members. A total of eight outer iterations and $T=30$ inner IEKF/IEKF-SL iterations were run. In order to demonstrate the effectiveness of these approximate credible intervals, we also present approximate credible intervals based on VIAS \cite{agrawal2021variational} in Figure \ref{EXAMPLE1_1}. The parameters of VIAS, were set to be $b = 0.1$ and $s = -0.495,$ with a total of 40 iterations. 

From Figure \ref{EXAMPLE1_1}, we can observe the effectiveness of our regularized methods in estimating the true parameter. The proposed $\ell_{0.5}$-IEKF/IEKF-SL clearly outperforms the vanilla IEKF/IEKF-SL in the estimation task. It is shown to be competitive with LASSO and IAS, which are only applicable in linear settings. Although the convexity of \eqref{sec2:objective} is not guaranteed for $p = 0.5$, the numerical result in Figure \ref{EXAMPLE1_1} demonstrates successful regularization. In terms of uncertainty quantification, $\ell_{0.5}$-IEKF/IEKF-SL showed comparable performance in constructing approximate credible intervals to the recently proposed VIAS.
These results clearly show that $\ell_{0.5}$-IEKF/IEKF-SL can preserve desirable properties of ensemble-based derivative-free optimization methods and IAS. To further investigate the convergence and the regularization effect of the parameter $r$, we also conducted extensive simulations using a less severely underdetermined system with $y\in \R^{100}$ to avoid possible instabilities caused by a small number of observations. To evaluate the regularization effect, we varied $r$ values ranging from 0.1 to 2 and computed the $\ell_2$-norm of the $\ell_1$-IEKF/IEKF-SL estimates whose indices correspond to entries off the support of the true signal. 

\begin{figure}[H]
\centering
\includegraphics[height = 0.25\textwidth]{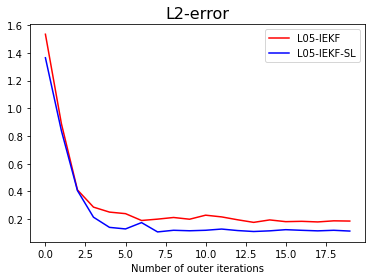}
\includegraphics[height = 0.25\textwidth]{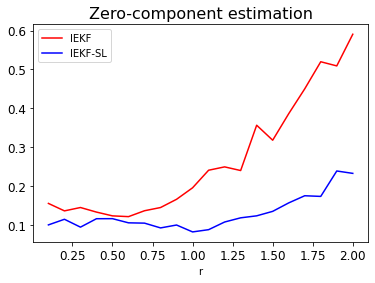}
\caption{Left: $\ell_2$-convergence comparison. Right: regularization effect of $r$. }
\label{EXAMPLE1_2}
\end{figure}
The left plot in Figure \ref{EXAMPLE1_2} shows the $\ell_2$ norm of the error $|u^\ell - u^*|$ between the true parameter $u^*$ and the estimates obtained by $\ell_{1}$-IEKF/IEKF-SL over 30 outer iterations. From the plot, $\ell_{1}$-IEKF-SL is more accurate than $\ell_{1}$-IEKF. The right plot in Figure \ref{EXAMPLE1_2} clearly demonstrates stronger regularization at small $r$ values, as expected.

\subsection{Nonlinear Inverse Problem with Explicit Forward Map}\label{ssec:nonlinearexplicitexample}
Here we study a nonlinear inverse problem introduced in \cite{kabanikhin2008definitions}, which has a closed-form forward map that satisfies the condition in Proposition \ref{sec4:convexity_prop}. Our results demonstrate that strongly sparsity-promoting regularization techniques, for which convexity of the objective is lost, provide more accurate reconstruction.
\subsubsection{Setting}
Consider the following first-order partial differential equation (PDE)
\begin{align}\label{eq:firstorederPDE}
\begin{split}
    \partial_{x_1} v - \partial_{x_2} v - u(x_1)v &= 0, \quad \quad \quad  (x_1,x_2) \in (0,1) \times (0,1), \\
    v(x_1,0) &= \phi(x_1), \quad\quad\quad \,\,  x_1 \in [0,1].
\end{split}
\end{align}
If $u$ is continuous and $\phi$ is continuously differentiable, then \eqref{eq:firstorederPDE} admits the solution 
$$
v(x_1,x_2) = \phi(x_1+x_2) \, \text{exp}\left(\int_{x_1+x_2}^{x_1} u(z)dz \right), \quad \quad (x_1,x_2) \in [0,1] \times [0,1].
$$
For a given $\phi$, the data $y$ is sampled according to
$$
y(x_1, x_2) = v(x_1, x_2) + \eps, \quad \eps \sim \Nc(0, 0.1^2), \quad (x_1,x_2) \in [0,1] \times [0,1].
$$
%In other words, data are comprised of point-wise evaluations of the aforementioned first-order PDE solution with noise added. 
The domain of interest, $[0,1] \times [0,1]$, was discretized with a $21 \times 21$ uniform grid. The solution of the PDE was observed on the grid. Our goal is to recover the function $u$ given the data $y \in \mathbb{R}^{21 \times 21}$. We further assume that $u$ admits a representation
$$
u(x) = \sum_{j=1}^{30}u_j \sin(j\pi x) +  \sum_{j=1}^{30}\tilde u_j \cos(j\pi x), ~x \in [0,1],
$$
which reduces the function recovery problem to a parameter estimation problem. Therefore, for each $(x_1, x_2) \in [0,1] \times [0,1]$, the forward map is given by
$$
\mathcal{G}:(u_j, \tilde u_j)_{j=1}^{30} \in \mathbb{R}^{60} \mapsto v \in \mathbb{R}^{441},
$$ 
which is convex/concave if the function $\phi$ is positive/negative. Hence, from the proposition \eqref{sec4:convexity_prop}, the objective function given in \eqref{sec2:objective} is convex for $r \ge 1$. The following simulations illustrate the effectiveness of the proposed methodology for $r = \frac{1}{3}$ ---which amounts to imposing an $\ell_{0.5}$-regularization--- despite the possible loss of convexity for $r \in (0,1)$. For the numerical simulation, we set $\phi(x) = \cos(x)$ and  assumed that only three components of each $u_j$ and $\tilde u_j$ are nonzero. Specifically, we set
$$
u(x) = 1.2\bigl(\sin(\pi x) + \sin(3\pi x) - \sin(6\pi x) - \cos(3\pi x)\bigr) - 0.6\bigl(\cos(\pi x) - \cos(6\pi x)\bigr).
$$

\subsubsection{Numerical Results}
We compared the performance of IEKF, IEKF-SL, $\ell_{0.5}$-IEKF, $\ell_{0.5}$-IEKF-SL, $\ell_1$-IEKF, and $\ell_1$-IEKF-SL in terms of their accuracy. For all of our methods, we used $N=100$ ensemble members with three outer regularization iterations and $T=20$ inner iterations of ensemble Kalman methods. The initial ensemble members $u_0^{(n)}$ were sampled from a centered Gaussian with covariance matrix $0.04 I_{60}$. We compare the reconstructions to the true target function $u$.

In both Figures \ref{EXAMPLE2_1} and \ref{EXAMPLE2_2}, the top and bottom rows correspond to $\ell_{1}$ and $\ell_{0.5}$-regularization,  respectively. The blue curves represent our function recovery based on the corresponding ensemble Kalman method. The shaded regions represent elementwise 2.5/97.5 percentile values of the recovery results. The recovery improves with additional outer iterations. As in the linear example, $\ell_{0.5}$-regularization worked as effectively as, or better than, $\ell_{1}$-regularization. In particular, $\ell_{0.5}$-IEKF-SL recovered  the true function almost perfectly after three outer iterations. 

\begin{minipage}{0.45\textwidth}
\begin{table}[H]
\begin{center}
\hspace{-1cm}
\begin{tabular}{ |c|c|c|c|c| } 
\hline
 & \multicolumn{3}{|c|}{$\#$ of outer iteration} \\
\hline
Method & 0th & 1st & 3rd \\
\hline
$\ell_{1}$-IEKF & 0.238 & 0.168 & 0.094  \\
%\hline
$\ell_{0.5}$-IEKF & 0.238 & 0.148 & 0.057 \\ 
%\hline
$\ell_{1}$-IEKF-SL & 0.205 & 0.141 & 0.067  \\
%\hline
$\ell_{0.5}$-IEKF-SL & 0.205 & 0.132 & 0.037  \\
\hline
\end{tabular}
\end{center}
\captionsetup{width=.8\linewidth}\caption{$\ell_2$-error between parameter estimate and true value.}
\label{table22}
\end{table}
\end{minipage}
\begin{minipage}{0.45\textwidth}
\begin{table}[H]
\begin{center}
\hspace{-0.4cm}
\begin{tabular}{ |c|c|c|c|c| } 
\hline
 & \multicolumn{3}{|c|}{$\#$ of outer iteration} \\
\hline
Method & 0th & 1st & 3rd \\
\hline
$\ell_{1}$-IEKF & 3.678 & 5.371 & 4.228  \\
%\hline
$\ell_{0.5}$-IEKF & 3.678 & 4.945 & 2.959  \\ 
%\hline
$\ell_{1}$-IEKF-SL & 4.660 & 3.649 & 2.898  \\
%\hline
$\ell_{0.5}$-IEKF-SL & 4.660 & 3.060 & 2.400 \\
\hline
\end{tabular}
\end{center}
\captionsetup{width=.8\linewidth}\caption{Average width of the credible intervals for recovery.}
\label{table21}
\end{table}
\end{minipage}

Table \ref{table22} shows the $\ell_2$-norm error between the parameter estimate and the truth. The results demonstrate the effectiveness of sparsity-promoting regularization. In all four methods, the $\ell_2$-error decreased with additional outer iterations. In both IEKF and IEKF-SL, $\ell_{0.5}$-regularization produced more accurate ($\ell_2$) recovery than $\ell_{1}$-regularization.

Table \ref{table21} contains the average widths of the credible intervals along the number of outer iterations. The widths of the credible intervals tend to decrease as more outer iterations are performed. This is expected since, as regularization effects accumulate, ensembles are more likely to center about their mean and credible bands around the target function become narrower. We can also see the stronger regularizing effect by comparing the width of the credible intervals corresponding to $\ell_{0.5}$-regularization and that of $\ell_{1}$-regularization.

\begin{figure}
\centering
\includegraphics[height=.25 \textwidth, width = \textwidth]{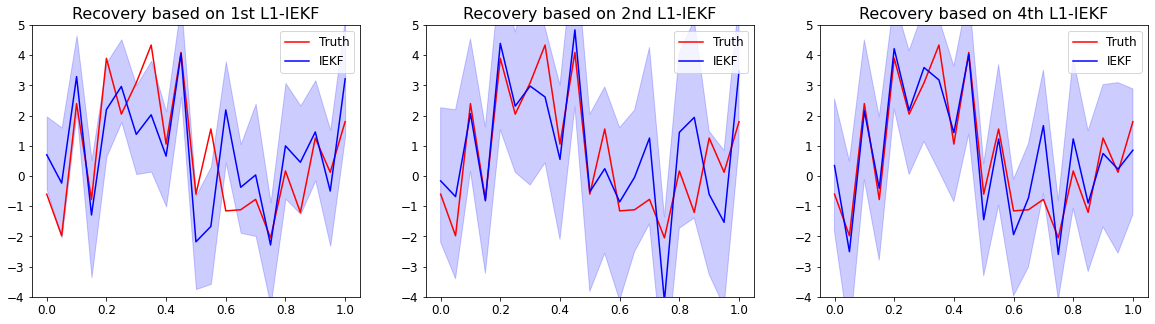}
\includegraphics[height=.25 \textwidth, width = \textwidth]{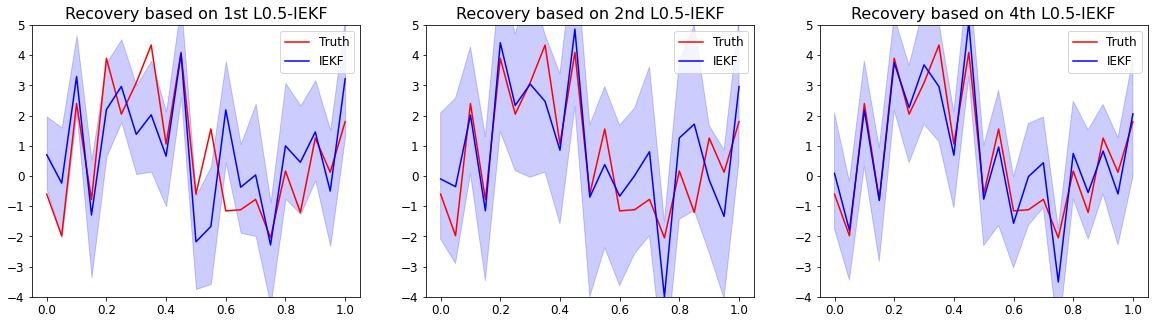}
\caption{Example in Subsection \ref{ssec:nonlinearexplicitexample}. Red: target function to recover. Blue: $\ell_p$-IEKF recovery. Top row: $\ell_{1}$-IEKF. Bottom row: $\ell_{0.5}$-IEKF. Left column: vanilla IEKF. Middle column: $\ell_p$-IEKF after one outer iteration. Right column: $\ell_p$-IEKF after three outer iterations. Shaded: 2.5/97.5 percentile of the recovery.}
\label{EXAMPLE2_1}
\end{figure}
\begin{figure}
\centering
\includegraphics[height=.25 \textwidth, width = \textwidth]{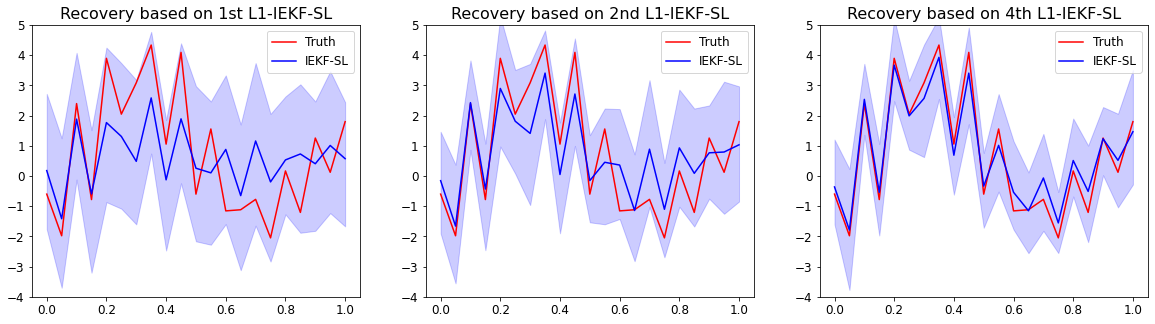}
\includegraphics[height=.25 \textwidth,width = \textwidth]{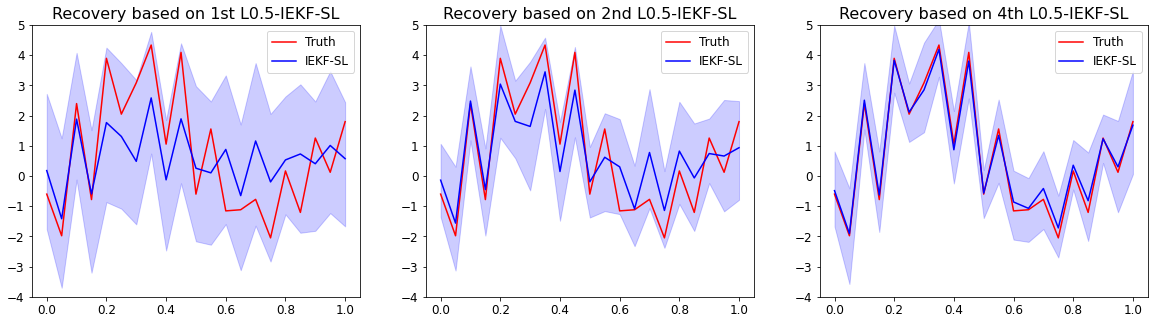}
\caption{Example in Subsection \ref{ssec:nonlinearexplicitexample}. Red: target function to recover. Blue: $\ell_p$-IEKF-SL recovery. Top row: $\ell_{1}$-IEKF-SL. Bottom row: $\ell_{0.5}$-IEKF-SL. Left column: vanilla IEKF-SL. Middle column: $\ell_p$-IEKF-SL after one outer iteration. Right column: $\ell_p$-IEKF-SL after three outer iterations. Shaded: 2.5/97.5 percentile of the recovery.}
\label{EXAMPLE2_2}
\end{figure}

%\begin{figure}[H]
%\centering
%\includegraphics[height=.3 \textwidth, width = \textwidth]{EXAMPLE2_TEKI_L05_124}
%\includegraphics[height=.3 \textwidth, width = \textwidth]{EXAMPLE2_TEKI_L1_124}
%\includegraphics[height=.3 \textwidth, width = \textwidth]{SISC-HierarchicalEnsembleKalman/EXAMPLE2_TEKI_L05_COEF124.png}
%\includegraphics[height=.3 \textwidth, width = \textwidth]{SISC-HierarchicalEnsembleKalman/EXAMPLE2_TEKI_L1_COEF124.png}
%\label{523}
%\caption{Recovery of function for the nonlinear toy example based on $\ell_{0.5}$-TEKI/$\ell_{1}$-TEKI. Red: Target function to recover, Blue: Recovered Function.}
%\end{figure}

\subsection{2D-Elliptic Inverse Problem}\label{ssec:ellipticexample}
Finally, following \cite{lee2021lp} we consider a two-dimensional elliptic inverse problem under a sparsity assumption. We show that our methodology can achieve accurate reconstructions with $\ell_p$-regularization with $p=0.5,$ while the approach in  \cite{lee2021lp} could not be implemented with $p<0.7.$
\subsubsection{Setting}
Consider the elliptic PDE
$$
-\text{div} \bigl(e^{u(x)} \nabla v(x)\bigr) = f(x), ~x = (x_1, x_2) \in [0,1] \times [0,1]
$$
with boundary conditions
$$
v(x_1, 0) = 100, ~\frac{\partial v}{\partial x_1}(1, x_2) = 0,~ -e^{u(x)}\frac{\partial v}{\partial x_1}(0, x_2) = 500, ~
\frac{\partial v}{\partial x_2}(x_1, 1) = 0,
$$
and source term
$$
f(x) = f(x_1, x_2) = 
\begin{cases}
0 & \quad 0 \le x_2 \le \frac{4}{6},\\
137 & \quad \frac{4}{6} < x_2 \le \frac{5}{6},\\
274 & \quad \frac{5}{6} < x_2 \le 1.
\end{cases}
$$
The equation is discretized in a uniform $15 \times 15$ grid in $[0,1] \times [0,1]$. Following \cite{lee2021lp}, we assumed that the log diffusion coefficient can be expressed as 
$$
u(x_1, x_2) = \sum_{i=0}^{19}\sum_{j=0}^{19} u_{ij} \phi_{ij}(x_1, x_2),
$$
where $\phi_{ij}(x_1, x_2) = \cos(i\pi x_1)\cos(j \pi x_2)$. Using the first boundary condition, the forward map is given by
$$
\mathcal{G}: \{u_{i,j}\}_{i,j=0}^{19} \in \mathbb{R}^{400} \mapsto v \in \mathbb{R}^{14 \times 15},
$$
which we implemented using the five-point stencil finite-difference method. To ensure sparsity, only six of the 400 components of $\{u_{i,j}\}_{i,j=0}^{19}$ were chosen to be nonzero. We aim to recover $\{u_{i,j}\}_{i,j=0}^{19} \in \mathbb{R}^{400}$ from the data
$$
y(x_1, x_2) = \mathcal{G}\bigl(u(x_1, x_2)\bigr) + \eps = v(x_1, x_2) + \eps, ~\eps \sim \Nc(0, 0.1^2).
$$

\subsubsection{Numerical Results}
We compare the performance of IEKF, IEKF-SL, $\ell_{0.5}$-IEKF, $\ell_{0.5}$-IEKF-SL, $\ell_1$-IEKF, and $\ell_1$-IEKF-SL in terms of their parameter estimation accuracy. For all our methods, we used $N=400$ ensemble members with six outer regularization iterations and $T=30$ inner iterations. The initial ensemble $u_0^{(n)}$ was sampled from a centered Gaussian with covariance matrix $0.1 I_{400}$. We present comparisons to the true parameters $\{u_{i,j}\}_{i,j=0}^{19}.$

In Figures \ref{EXAMPLE3_1} and \ref{EXAMPLE3_2}, the top and bottom rows respectively correspond to $\ell_{1}$ and $\ell_{0.5}$-regularization. We also provide approximate credible intervals, constructed from elementwise 2.5/97.5 percentile values of the empirical distribution of the ensemble. In both figures, we observe that estimates improve with more outer iterations. In addition, $\ell_{0.5}$-regularization acted more strongly off the support than $\ell_1$-regularization. Within three outer regularization iterations, both $\ell_{0.5}$-IEKF/IEKF-SL yielded parameter estimates very close to the true value. 

As in the previous numerical example, Table \ref{table32} provides the $\ell_2$-norm errors. In all cases, errors decreased with the additional outer iterations. Note that $\ell_{0.5}$-regularization is an order of magnitude more accurate than $\ell_{1}$-regularization whether using IEKF or IEKF-SL. Table \ref{table31} shows that the length of approximate credible intervals decreased with the number of outer iterations. After only six iterations, credible intervals produced by $\ell_{0.5}$-regularization are an order of magnitude or more smaller than those produced by $\ell_{1}$-regularization.

\begin{minipage}{0.45\textwidth}
\begin{table}[H]
\begin{center}
\hspace{-1.2cm}
\begin{tabular}{ |c|c|c|c|c| } 
\hline
 & \multicolumn{3}{|c|}{$\#$ of outer iteration} \\
\hline
Method & 0th & 3rd & 6th \\
\hline
$\ell_{1}$-IEKF & 0.030 & 0.014 & 0.012  \\
%\hline
$\ell_{0.5}$-IEKF & 0.030 & 0.004 & 0.002 \\ 
%\hline
$\ell_{1}$-IEKF-SL & 0.030 & 0.020 & 0.014  \\
%\hline
$\ell_{0.5}$-IEKF-SL & 0.030 & 0.008 & 0.007  \\
\hline
\end{tabular}
\end{center}
\captionsetup{width=.8\linewidth}\caption{$\ell_2$-error between parameter estimate and true value.}
\label{table32}
\end{table}
\end{minipage}
\begin{minipage}{0.45\textwidth}
\begin{table}[H]
\begin{center}
\hspace{-0.4cm}
\begin{tabular}{ |c|c|c|c|c| } 
\hline
 & \multicolumn{3}{|c|}{$\#$ of outer iteration} \\
\hline
Method & 0th & 3rd & 6th \\
\hline
$\ell_{1}$-IEKF & 1.211 & 0.198 & 0.211  \\
%\hline
$\ell_{0.5}$-IEKF & 1.211 & 0.055 & 0.017  \\ 
%\hline
$\ell_{1}$-IEKF-SL & 1.399 & 0.259 & 0.169  \\
%\hline
$\ell_{0.5}$-IEKF-SL & 1.399 & 0.041 & 0.007  \\
\hline
\end{tabular}
\end{center}
\captionsetup{width=.8\linewidth}\caption{Average width of the credible intervals for recovery.}
\label{table31}
\end{table}
\end{minipage}

%\begin{table}[H]
%\begin{center}
%\begin{tabular}{ |c|c|c|c|c| } 
%\hline
%Method & No outer & Three outer & Six outer \\
%\hline
%$\ell_{1}$-IEKF & 1.211 & 0.198 & 0.211  \\
%\hline
%$\ell_{0.5}$-IEKF & 1.211 & 0.055 & 0.017  \\ 
%\hline
%$\ell_{1}$-IEKF-SL & 1.399 & 0.259 & 0.169  \\
%\hline
%$\ell_{0.5}$-IEKF-SL & 1.399 & 0.041 & 0.007  \\
%\hline
%\end{tabular}
%\end{center}
%\caption{Average width of the credible intervals for parameter estimate}
%\label{table31}
%\end{table}

%\begin{table}[H]
%\begin{center}
%\begin{tabular}{ |c|c|c|c|c| } 
%\hline
%Method & No outer & Three outer & Six outer \\
%\hline
%$\ell_{1}$-IEKF & 0.592 & 0.288 & 0.241  \\
%\hline
%$\ell_{0.5}$-IEKF & 0.592 & 0.072 & 0.036 \\ 
%\hline
%$\ell_{1}$-IEKF-SL & 0.600 & 0.405 & 0.277  \\
%\hline
%$\ell_{0.5}$-IEKF-SL & 0.600 & 0.152 & 0.140  \\
%\hline
%\end{tabular}
%\end{center}
%\caption{$\ell_2$-norm of the difference between parameter estimate and true value}
%\label{table32}
%\end{table}

\begin{figure}
\centering
\includegraphics[height=.275\textwidth]{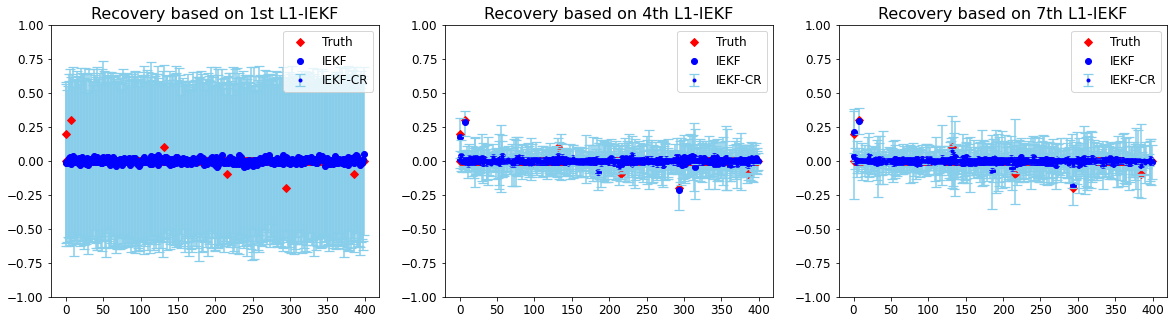}
\includegraphics[height=.275\textwidth]{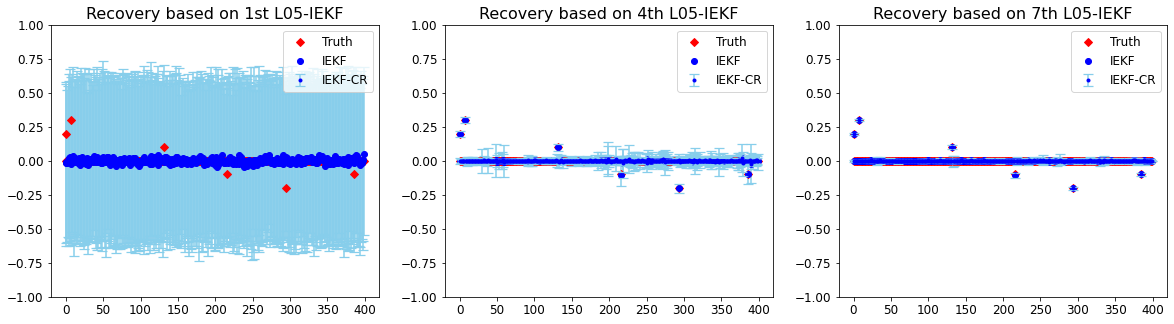}
\caption{Parameter recovery for 2D-elliptic inverse problem based on $\ell_{1}$/$\ell_{0.5}$-IEKF. Red: Truth. Blue: $\ell_p$-IEKF estimate. Left column: vanilla IEKF. Middle column: $\ell_p$-IEKF after three outer iterations. Right column: $\ell_p$-IEKF after six outer iterations. Shaded: elementwise 2.5/97.5 percentile for parameter estimate.}
\label{EXAMPLE3_1}
\end{figure}

\begin{figure}
\centering
\includegraphics[height=.275\textwidth]{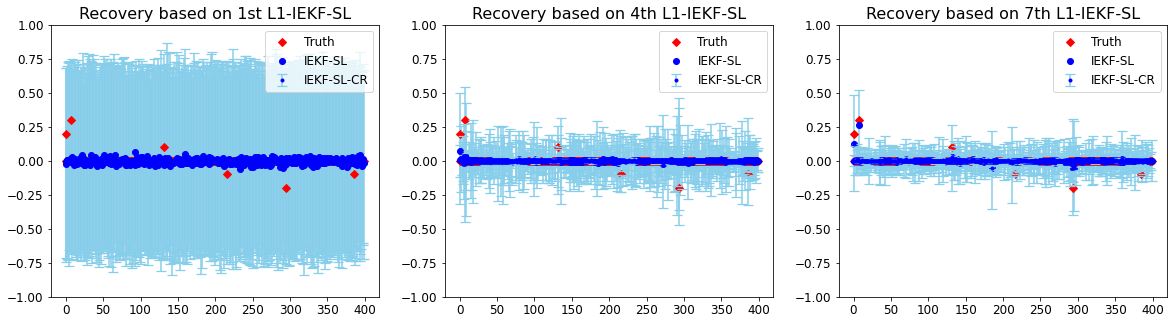}
\includegraphics[height=.275\textwidth]{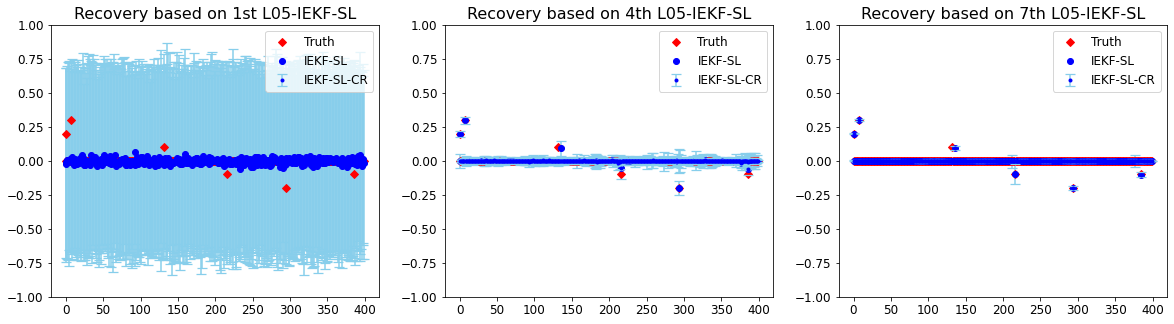}
\caption{Parameter recovery for 2D-elliptic inverse problem based on $\ell_{1}$/$\ell_{0.5}$-IEKF-SL. Red: Truth. Blue: $\ell_p$-IEKF-SL estimate. Left column: vanilla IEKF-SL. Middle column: $\ell_p$-IEKF-SL after three outer iterations. Right column: $\ell_p$-IEKF-SL after six outer iterations. Shaded: elementwise 2.5/97.5 percentile for parameter estimate.}
\label{EXAMPLE3_2}
\end{figure}

\section{Conclusion}\label{sec:Conclusions}
This paper introduced a flexible computational framework to incorporate a wide range of regularization techniques in ensemble Kalman methods. We have adopted a hierarchical Bayesian perspective to derive our methodology and shown that suitable choices of hyperparameters yield sparsity-promoting regularization. The effectiveness of our procedure was demonstrated in three numerical examples. While we have focused on  sparsity-promoting $\ell_p$-penalties, our framework extends beyond sparse models. In particular, heavy-tailed Student prior regularization and relaxed $\ell_p$-penalties with $r \beta > 3/2$ could be considered for applications in nonlinear regression and in learning dynamical systems from time-averaged data \cite{schneider2020imposing}. Finally, this paper focused on ensemble Kalman methods for inverse problems; future work will investigate regularization of ensemble Kalman filters \cite{sanzstuarttaeb,chen2021auto} in  data assimilation.

\section*{Acknowledgments}
The authors are thankful to Y. Chen and O. Ghattas for their generous feedback on a previous version of this manuscript. 
DSA is grateful for the support of NSF DMS-2027056, DOE DE-SC0022232, and a FBBVA start-up grant. HK was partially supported by DMS-2027056.

\bibliographystyle{plain} 
\bibliography{references}
\end{document}